\documentclass[11pt, a4paper,reqno]{amsart}

\usepackage{color} \definecolor{bleu_sombre}{rgb}{0,0,0.6}  \definecolor{rouge_sombre}{rgb}{0.8,0,0}\definecolor{vert_sombre}{rgb}{0,0.6,0}
\usepackage[plainpages=false,colorlinks,linkcolor=bleu_sombre,citecolor=rouge_sombre,urlcolor=vert_sombre,breaklinks]{hyperref}

\usepackage[english]{babel}

\usepackage{amsmath,amssymb,amsthm,graphicx,amsfonts,url,color,enumerate,dsfont,stmaryrd,mathabx}

\theoremstyle{plain}
\newtheorem{theorem}{{Theorem}}[section] 
\newtheorem*{theorem*}{{Theorem}}
\newtheorem{proposition}[theorem]{Proposition}
\newtheorem*{proposition*}{Proposition}
\newtheorem{corollary}[theorem]{Corollary}
\newtheorem*{corollary*}{Corollary}
\newtheorem{lemma}[theorem]{Lemma}
\newtheorem*{lemma*}{Lemma}

\theoremstyle{definition}

\newtheorem*{definition*}{Definition}

\theoremstyle{remark}
\newtheorem{remark}[theorem]{Remark}

\makeatletter

\@addtoreset{equation}{section}  
\makeatother

\newcommand{\abs}[1]{\left\vert #1\right\vert}      
\newcommand{\nr}[1]{\left\Vert #1\right\Vert}         
\newcommand{\innp}[2]{\left< #1 , #2 \right>}          
\newcommand{\pppg}[1] {\left< #1 \right>}
\newcommand{\set}[1]{\left\{ #1 \right\}}		
\newcommand{\Ii}[2] {\left\{ #1,\dots,#2 \right\}}
\renewcommand{\leq}{\leqslant}	\renewcommand{\geq}{\geqslant}
\renewcommand{\bar}[1]{\overline{#1}}
\newcommand{\ie}{{\it{i.e. }}}
\newcommand{\inv}{^{-1}}
\newcommand{\bigo}[2]{\mathop{O}\limits_{#1 \to #2}}

\newcommand{\st}{\,:\,}
\newcommand{\restr}[2]{\left.#1\right|_{#2}}         
\renewcommand{\Re}{\mathsf{Re}}        
\renewcommand{\Im}{\mathsf{Im}}  
\newcommand{\trsp}{^{\intercal}}

\newcommand{\eqv}{\Longleftrightarrow}

\newcommand{\Dom}{\mathsf{Dom}}
\newcommand{\Sp}{\mathsf{Sp}}
\renewcommand{\ker}{\mathsf{ker}} 
\newcommand{\Ran}{\mathsf{Ran}}
\newcommand{\Id}{\mathsf{Id}} 

\newcommand{\divg}{\mathop{\rm{div}}\nolimits}
\newcommand{\loc}{{\mathsf{loc}}}

\newcommand{\R}{\mathbb{R}}		\newcommand{\C}{\mathbb{C}}
\newcommand{\N}{\mathbb{N}}	\newcommand{\Z}{\mathbb{Z}}	
    \newcommand{\T}{\mathbb{T}}

\renewcommand{\a}{\alpha}\renewcommand{\b}{\beta}\newcommand{\g}{\gamma}\newcommand{\G}{\Gamma}\renewcommand{\d}{\delta}\newcommand{\D}{\Delta}\newcommand{\e}{\varepsilon}\newcommand{\z}{\zeta} \newcommand{\y}{\eta}\renewcommand{\th}{\theta}\newcommand{\Th}{\Theta}\renewcommand{\k}{\kappa}\renewcommand{\l}{\lambda}\renewcommand{\L}{\Lambda}\newcommand{\m}{\mu}\newcommand{\n}{\nu}\newcommand{\x}{\xi}\newcommand{\s}{\sigma}\renewcommand{\t}{\tau}\newcommand{\f}{\varphi}\newcommand{\vf}{\phi}\newcommand{\h}{\chi}\newcommand{\p}{\psi}

\newcommand{\Ac}{{\mathcal A}}\newcommand{\Bc}{{\mathcal B}}\newcommand{\Dc}{{\mathcal D}}\newcommand{\Hc}{{\mathcal H}}\newcommand{\Kc}{{\mathcal K}}\newcommand{\Lc}{{\mathcal L}}\newcommand{\Oc}{{\mathcal O}}\newcommand{\Pc}{{\mathcal P}}\newcommand{\Rc}{{\mathcal R}}\newcommand{\Sc}{{\mathcal S}}\newcommand{\Tc}{{\mathcal T}}\newcommand{\Uc}{{\mathcal U}}\newcommand{\Vc}{{\mathcal V}}

\newcommand{\stepp}{\noindent {\bf $\bullet$}\quad }

\usepackage{mathrsfs}

\begin{document}

\newcommand{\HH}{\mathcal H}\newcommand{\LL}{\mathcal L}\newcommand{\EE}{\mathscr E}

\newcommand{\AAc}{\Ac} \newcommand{\tAAc}{\widetilde \Ac} \newcommand{\AAp}{\Ac_\mathbf{p}} \newcommand{\tAAp}{\widetilde \Ac_\mathbf{p}} \newcommand{\AAs}{\Ac_\s} \newcommand{\AAss}{\Ac_\s^*} \newcommand{\AAso}{\Ac_0} \newcommand{\AAsso}{\Ac_0^*} 

\newcommand{\Pg}{P}\newcommand{\PG}{P_G}\newcommand{\Pgper}{\tilde P_{\mathbf{p}}}
\newcommand{\Pp}{P_{\mathbf{p}}}\newcommand{\Pm}{P_{\mathbf{h}}}\newcommand{\Po}{P_{0}}
\newcommand{\Psp}{P_{\mathbf{p}}^\s}
\newcommand{\Psnp}{P_{\mathbf{p}}^{\s_n}}\newcommand{\Pop}{P_{\mathbf{p}}^{0}}
\newcommand{\Gp}{G_{\mathbf{p}}}\newcommand{\Go}{G_{0}}\newcommand{\Gm}{G_{\mathbf{h}}}
\newcommand{\ap}{a_{\mathbf{p}}}\newcommand{\aps}{a_{\mathbf{p}}^*}\newcommand{\ao}{a_{0}}
\renewcommand{\wp}{w_{\mathbf{p}}}\newcommand{\wps}{w_{\mathbf{p}}^*}\newcommand{\wo}{w_{0}}
\newcommand{\bm}{b_{\mathbf{h}}}\newcommand{\bp}{b_{\mathbf{p}}}
\newcommand{\Rp}{R_{\mathbf{p}}} 
\newcommand{\Ilow}{I_{\low}}
\newcommand{\tIlow}{\widetilde I_{\mathbf{p}}}
\newcommand{\Ilowp}{I_{\mathbf{p}}}
\newcommand{\Ps}{\Pi_\s}\newcommand{\Pso}{\Pi_0}
\newcommand{\Up}{U_{\mathbf{p}}} \newcommand{\upp}{u_{\mathbf{p}}} \newcommand{\um}{u_{\mathbf{h}}}
\newcommand{\ls}{\l_\s}\newcommand{\fs}{\f_\s} \newcommand{\fo}{\f_0} 
\newcommand{\low}{{\rm{low}}}
\newcommand{\high}{{\rm{high}}}
\newcommand{\Ulow}{U_{\low}^{\e}} \newcommand{\Uhigh}{U_{\high}^{\e}}
\newcommand{\per}{\#}
\newcommand{\Thbxbt}{\Th^{\b_x}_{\b_t}} \newcommand{\Thbxbts}{\Th^{\b_x}_{\b_t}(\s)}
\newcommand{\Thbxo}{\Th^{\b_x}_{0}}\newcommand{\Thoo}{\Th_{0}^{0}}
\newcommand{\Thuo}{\Th_{0}^{1}} \newcommand{\Thobt}{\Th^0_{\b_t}}
\newcommand{\tThbx}{\widetilde \Th_{\tilde \b_x}}\newcommand{\tThbxs}{\widetilde \Th_{\tilde \b_x}(\s)}
\newcommand{\Tha}{\Th^\a_0}\newcommand{\tTha}{\widetilde \Th_{\tilde \a}}
\newcommand{\Fzper}{F_1^{\mathbf{p}}(z)}
\newcommand{\Tooper}{\Tc_{0,0}^{\mathbf{p}}(t)}
\newcommand{\rr}{r}
\renewcommand{\ggg}{\g}
\newcommand{\To}{\Tc_0} \newcommand{\tTo}{\tilde \Tc_0}\newcommand{\Toper}{\Tc_{0,\per}} \newcommand{\Tinf}{\Tc_\infty}

\title[The asymptotically periodic damped wave equation]{Energy decay and diffusion phenomenon for the asymptotically periodic damped wave equation}

\author{Romain Joly}
\address[R. Joly]{Institut Fourier - UMR5582
CNRS/Universit\'e Grenoble Alpes - 100, rue des Maths -
CS 40700
F-38058 Grenoble cedex 9, France}
\email{romain.joly@univ-grenoble-alpes.fr} 

\author{Julien Royer}
\address[J. Royer]{Institut de Math\'ematiques de Toulouse, Universit\'e Toulouse 3, 118 route de Narbonne - F31062 Toulouse c\'edex 9, France}
\email{julien.royer@math.univ-toulouse.fr}


\begin{abstract}
We prove local and global energy decay for the asymptotically periodic damped wave equation on the Euclidean space. Since the behavior of high frequencies is already mostly understood, this paper is mainly about the contribution of low frequencies. We show in particular that the damped wave behaves like a solution of a heat equation which depends on the H-limit of the metric and the mean value of the absorption index.
\end{abstract}

\maketitle

\section{Introduction and statement of the main results}

In this paper we are interested in the asymptotic behavior for large times of the damped wave equation in an asymptotically periodic setting in $\R^d$, $d \geq 1$. In particular, the damping is effective at infinity but it is not assumed to be greater than a positive constant outside some compact subset of $\R^d$. Our original motivation is the local energy decay. We also obtain some results for the global energy. However, because of the contribution of low frequencies, there is no exponential decay for the corresponding semigroup, even under the usual Geometric Control Condition. More precisely, we will prove that the contribution of low frequencies behaves like a solution of an explicit heat equation. This will explain the rate of decay for the local energy decay.

\subsection{The damped wave equation in an asymptotically periodic setting}

We consider on $\R^d$ the damped wave equation
\begin{equation} \label{wave}
\begin{cases}
\partial_t ^2 u  + \Pg u + a(x) \partial_t u = 0 & \text{on } \R_+ \times  \R^d,\\
\restr{(u,\partial_t u)}{t=0} = (u_0,u_1) & \text{on } \R^d,
\end{cases}
\end{equation}
where $(u_0,u_1) \in H^1(\R^d) \times L^2(\R^d)$. 

The function $a$ is the absorption index. It is bounded, continuous, and takes non-negative values.

The operator $\Pg$ is a general Laplace operator. More explicitely, we consider a metric $G(x) = (G_{j,k}(x))_{1\leq j,k\leq d}$ on $\R^d$ and a positive function $w$ such that, for some $G_{\max} \geq G_{\min} > 0$ and $w_{\max} \geq w_{\min} > 0$ and for all $x \in \R^d$ and $\x \in \R^d$,
\begin{equation} \label{hyp-GbC}
G_{\min} \abs \x^2 \leq \innp{G(x) \x}{\x} \leq G_{\max} \abs \x^2 \quad \text{and} \quad w_{\min} \leq w(x) \leq w_{\max}.
\end{equation}
We also assume that $G$ and $w$ are smooth with bounded derivatives. Then we set
\begin{equation} \label{def-P}
\Pg := -\frac 1 {w(x)} \divg G(x) \nabla.
\end{equation}
This includes in particular the case of the standard Laplace operator (with $G(x) = \Id$ and $w(x) = 1$), a Laplacian in divergence form (with $w(x) = 1$) or the Laplacian associated with a metric $g(x)$ (with $w(x) = \det(g(x))^{\frac 12}$ and $G(x) = \det(g(x))^{\frac 12} g(x)\inv$).\\

The purpose of this paper is to consider the case where $G$, $w$ and $a$ are asymptotically periodic. This means that we can write  
\[
G(x) = \Gp(x) + \Go(x), \quad w(x) = \wp(x) + \wo(x)  \quad \text{and} \quad a(x) = \ap(x) + \ao(x),
\]
where $\Gp$, $\wp$ and $\ap$ are $\Z^d$-periodic and $\Go$, $\wo$ and $\ao$ go to 0 at infinity. More precisely, we assume that there exist $\rho_G, \rho_a > 0$ and $C_G,C_a \geq 0$ such that 
\begin{equation} \label{hyp-vanishing}
\abs{\Go(x)} \leq C_G \pppg x^{-\rho_G} \quad \text{and} \quad \abs{\wo(x)} + \abs{\ao(x)} \leq C_a \pppg x^{-\rho_a},
\end{equation}
where $\pppg x$ stands for $\big( 1 + \abs x^2 \big)^{\frac 12}$. 
The periodic part $\ap$ of the absorption index is allowed to vanish but it is not identically zero, so that the damping is effective at infinity. 
Notice that if $G(x)$ and $w(x)$ are periodic and $a(x)$ is constant, then we recover the setting of \cite{OrivePaZu01}.\\

Let $u$ be a solution of \eqref{wave}. We can check that if $a = 0$ then the energy 
\begin{equation} \label{def-E}
E(t) := \int_{\R^d} \left(w(x) \abs{\partial_t u(t,x)}^2 + G(x) \nabla u(t,x) \cdot \nabla \bar u(t,x) \right) \, dx
\end{equation}
is constant. However, with the damping this is a non-increasing function of time. More precisely, for $t_1 \leq t_2$ we have
\begin{equation} \label{easy-decay}
E(t_2) - E(t_1) = - 2 \int_{t_1}^{t_2} \int_{\R^d} a(x) \abs{\partial_t u(t,x)}^2 w(x) \, dx \, dt \leq 0.
\end{equation}
Our purpose in this paper is to say more about the decay of this quantity. We are also interested in the decay of the local energy 
\[
E_\d(t) := \int_{\R^d} \pppg x^{-2\d} \left( w(x)\abs{\partial_t u(t,x)}^2 + G(x) \nabla u(t,x) \cdot \nabla \bar u(t,x) \right) \, dx,
\]
where $\d \geq 0$.

\subsection{The geometric damping condition on classical trajectories}

The local energy decay for the wave equation in unbounded domains and the global energy decay for the damped wave equation in compact domains are two problems which have quite a long history.

In the first case the global energy is conserved but, at least for the free setting, the energy escapes to infinity. In perturbed settings, it is then important to know wether some energy can be trapped, to estimate the dependance of the decay of the local energy with respect to the initial condition, etc. We refer for instance to \cite{morawetzrs77,melrose79,burq98,bonyh12,bouclet11} for different results in various asymptotically free settings.

For the damped wave equation we really have a loss of energy. Then the goal of stabilisation results is to understand where the damping should be effective to make this energy go to 0 (with the same kind of questions about the rates of decay). We refer for instance to \cite{raucht74,bardoslr92,lebeau96,lebeaur97}. 

The behavior of the energy of a wave depends on its frequency. The main difficulties usually come from the contributions of high and low frequencies. It is now well known that for high frequencies the behavior of the wave depends on the geometry of the domain. More precisely, the wave basically propagates following the classical trajectories for the corresponding Hamiltonian problem. 
Then the local energy decays uniformly in unbounded domains if and only if all these trajectories go to infinity (this is the so-called non-trapping condition), while for the damped wave equation in compact domains, the global energy decays uniformly if and only if all the classical trajectories meet the damping region (this is the geometric control condition, G.C.C. for short). 
The problems with the contributions of low frequencies only appear in unbounded domains. The local energy for the contribution of low frequencies decays uniformly without assumption, but it can be slower than for high frequencies. Typically, for compactly supported perturbations of the free setting in even dimension, the local energy for the contribution of low frequencies decays like $t^{-2d}$, while the contribution of high frequencies decays faster than any power of $t$ under the non-trapping condition.\\

In this paper we analyse the local energy decay for damped wave equation in an unbounded domain. In this case the criterion for the contribution of high frequencies combines the non-trapping and the geometric control conditions: each bounded trajectories should either go through the damping region or escape to infinity. 

For a compactly supported or asymptotically vanishing damping, we recover with this assumption the same kind of results as for the undamped analog under the non-trapping condition. See \cite{alouik02,khenissi03, boucletr14, royer-dld-energy-space}. This is basically due to the fact that the part which escapes to infinity is no longer influenced by the damping and behaves as in the free case. In this kind of setting the trajectories at infinity never see the damping, so we cannot expect a global energy decay.\\

The situation is quite different when the damping is effective at infinity. In the asymptotically periodic case, we have at least the property that all the points in $\R^d$ are uniformly close to the damping region. 

For the contribution of high frequencies we will use the results of \cite{BurqJo}, where the damped Klein-Gordon equation is considered in a similar setting. We recall that the Klein-Gordon equation is analogous to the wave equation, except that the non-negative operator $P$ is replaced by $P + 1$. In this case there is no difficulty with the low frequencies (0 is no longer in the spectrum), but this does not make any significant difference for the contribution of high frequencies. So for high frequencies it is equivalent to look at the wave or at the Klein-Gordon equation.

Thus, we can first deduce from \cite{BurqJo} that we have at least a logarithmic decay with loss of regularity for the contribution of high frequencies . If $\Pg = -\D$ and $a$ is periodic, then by \cite{Wunsch} we obtain a polynomial decay (still with loss of regularity). The best decay is obtained when all the classical trajectories go uniformly through the damping. Since our main purpose is the analysis of the contribution of low frequencies, we assume that this is the case in this paper.  

For a more precise statement, we introduce on $\R^{2d} \simeq T^* \R^d$ the symbol 
\[
p : (x,\x) \mapsto \frac {\innp{G(x)\x}{\x}}{w(x)}
\]
and the corresponding classical flow: for $(x_0,\x_0) \in \R^{2d}$ we denote by $\vf^t(x_0,\x_0)$ the solution of the Hamiltonian problem 
\[
\begin{cases}
\frac d {dt} \vf^t(x_0,\x_0) =  \big(\nabla_\x p (\vf^t(x_0,\x_0)), -\nabla_x p (\vf^t(x_0,\x_0)) \big), \\
\vf^0(x_0,\x_0) = (x_0,\x_0).
\end{cases}
\]
We recall that $\vf^t(x_0,\x_0) = (x_0 + 2t\x_0, \x_0)$ if $\Pg = -\D$ and $\vf^t$ is the geodesic flow corresponding to the metric $g$ if $\Pg = -\D_g$. For a review about semiclassical analysis, we refer to \cite{zworski}. 

We assume that there exist $T > 0$ and $\a > 0$ such that 
\begin{equation} \label{hyp-GCC}
\forall (x_0,\x_0) \in p\inv(\set 1), \quad \int_0^T a \big( \vf^t(x_0,\x_0) \big) \, dt \geq \a,
\end{equation}
where we have extended $a$ to a function on $\R^{2d}$ which only depends on the first $d$ variables.
Under this assumption, we know from Theorem 1.2 in \cite{BurqJo} that the global (and therefore local) energy of the contribution of high frequencies decays uniformly (without loss of regularity) exponentially. Thus, in all the results of this paper, the restrictions in the rates of decay are due to the contributions of low frequencies.

\subsection{Energy decay for the damped wave equation in the periodic setting}

After multiplication by $w(x)$, the problem \eqref{wave} reads
\begin{equation} \label{wave-w}
\begin{cases}
w(x) \partial_t ^2 u + \PG u + b(x) \partial_t u = 0 & \text{on } \R_+ \times  \R^d,\\
\restr{(u,\partial_t u)}{t=0} = (u_0,u_1) & \text{on } \R^d,
\end{cases}
\end{equation}
where $b(x) := w(x) a(x)$ and $\PG$ is a Laplacian in divergence form:
\[
\PG := -\divg G(x) \nabla.
\]

We denote by $\Sc$ the Schwartz space of smooth functions whose derivatives decay faster than any polynomial at infinity. For $\d \in \R$ we denote by $L^{2,\d}(\R^d)$ the weighted space $L^2(\pppg x^{2\d} \, dx)$ and by $H^{k,\d}(\R^d)$, $k\in\N$, the corresponding Sobolev space. Then we set 
\[
\LL := L^2(\R^d) \times L^2(\R^d), \quad \HH := H^1(\R^d) \times L^2(\R^d) \quad \text{and} \quad \HH^\d := H^{1,\d}(\R^d) \times L^{2,\d}(\R^d).
\]

We begin with the purely periodic case. Thus, for $(u_0,u_1) \in \HH$ we first consider the problem 
\begin{equation} \label{wave-per}
\begin{cases}
\wp(x) \partial_t ^2 \upp  + \Pp \upp + \bp(x) \partial_t \upp = 0 & \text{on } \R_+ \times  \R^d,\\
\restr{(\upp,\partial_t \upp)}{t=0} = (u_0,u_1) & \text{on } \R^d,
\end{cases}
\end{equation}
where
\[
\Pp := - \divg \Gp(x) \nabla \quad \text{and} \quad \bp(x) := \wp(x) \ap(x).
\]
In the following result we describe the local and global energy decay for the solution of \eqref{wave-per}.

\begin{theorem}[Local and global energy decay in the periodic setting] \label{th-energy-decay-periodic}
Assume that the damping condition \eqref{hyp-GCC} holds. Let $s_1, s_2 \in \big[0,\frac d 2\big]$ and $\k > 1$. Let $s \in [0,1]$. Then there exists $C \geq 0$ such that for $t \geq 0$ and $U_0 = (u_0,u_1) \in \HH^{\k s_2+s}$ we have 
\begin{eqnarray*}
\nr{\upp(t)}_{L^{2,-\k s_1}} &\leq& C \pppg t^{- \frac {s_1 + s_2} 2} \nr{U_0}_{\HH^{\k s_2}},\\
\nr{\partial_t \upp(t)}_{L^{2,-\k s_1}} &\leq& C \pppg t^{-1 - \frac {s_1 + s_2} 2} \nr{U_0}_{\HH^{\k s_2}},\\
\nr{\nabla \upp(t)}_{L^{2,-\k s_1 -s}} &\leq& C \pppg t^{- \frac {1 + s} 2 - \frac {s_1 + s_2} 2} \nr{U_0}_{\HH^{\k s_2 + s}},
\end{eqnarray*}
where $\upp(t)$ is the solution of \eqref{wave-per}.
\end{theorem}

Notice that we give decay estimates for the energy of the wave (\ie for the time and spatial derivatives of the solution), but also for the solution itself.

We will see that these estimates are sharp. When $s_1 = s_2 = s = 0$, we obtain estimates for the global energy (notice, however, that in the right-hand side $\nr{U_0}_{\Hc}$ is not the initial energy, see Remark \eqref{rem-energy-space} below). When $s_1$ is positive, we are estimating the local energy (which decays faster than the global energy). On the other hand, the parameter $s_2$ measures the localization of the initial data. 
We notice that even the global energy decays faster if the initial data is assumed to be localized. Finally we observe that the spatial derivatives do not play the same role as the time derivative, which is unusual for a wave equation. However, if we can take $s = 1$ (this is the case if we are interested in the local energy decay for localized initial data) then we recover for the spatial derivatives the same estimates as for the time derivative.

\subsection{Comparison with the solution of a heat equation}

As mentioned above, the rates of decay in Theorem \ref{th-energy-decay-periodic} are not usual for a wave equation. This is due to the contribution of low frequencies, which under a strong damping behaves like a solution of a heat equation.

This phenomenon has already been observed in earlier papers. The simplest case is the standard wave equation with constant damping 
\begin{equation} \label{eq-wave-model}
\partial_t^2 u - \D u + \partial_t u = 0.
\end{equation}
The energy decay for the solutions of \eqref{eq-wave-model} has been first studied in \cite{Matsumura76}. More precise results have then be given in \cite{nishihara03,marcatin03, hosonoo04, narazaki04}. In these papers it is proved that a solution of \eqref{eq-wave-model} behaves for large times like a solution of the heat equation
\begin{equation} \label{eq-heat-model}
- \D v + \partial_t v = 0.
\end{equation}
This phenomenon can be understood as follows. Since G.C.C. is satisfied when $a \equiv 1$, the behavior of the wave for large times is governed by the contribution of low frequencies. But for very slowly oscillating solutions, we expect that the contribution of the term $\partial_t^2 u$ in \eqref{eq-wave-model} will be very small compared to $\partial_t u$, and then $u$ will look like a solution of \eqref{eq-heat-model}. 
The same phenomenon has been observed in an exterior domain (see \cite{Ikehata02} for a constant absorption index and \cite{AlouiIbKh15} for an absorption index equal to 1 outside some compact) and in a wave guide (see \cite{royer-diss-wave-guide} for a constant dissipation at the boundary and \cite{MallougRo} for an asymptotically constant absorption index). For a slowly decaying absorption index ($a(x) = \pppg x^{-\rho}$ with $\rho \in (0,1]$) we refer to \cite{TodorovaYo09,
IkehataToYo13, Wakasugi14} (we recall from \cite{royer-dld-energy-space} that if $a(x) \lesssim \pppg x^{-\rho}$ with $\rho > 1$ then we recover the behavior of the undamped wave equation). For the problem in an exterior domain with possibly slowly decaying damping, we refer to \cite{SobajimaWak16}. These questions are also of interest for the semilinear damped wave equation (see \cite{Wakasugi17} and references therein). Finally, results on an abstract setting can be found in \cite{Chill-Ha-04,Radu-To-Yo-11,Nishiyama16,Radu-To-Yo-16}.\\

The same phenomenon occurs in our periodic setting. We can be more precise than in Theorem \ref{th-energy-decay-periodic} and prove that our wave can indeed be written as the sum of the solution of some heat equation on $\R^d$ and a smaller term (in the sense that it decays faster when $t$ goes to $+\infty$). Notice that this problem has already been studied in \cite{OrivePaZu01} (see the discussion after Theorem \ref{th-heat}).

As already said, this diffusive phenomenon is due to the contribution of low frequencies. Assume (at least formally) that $u$ is a solution of \eqref{wave-w} oscillating at a frequency $\t$ with $\abs \t \ll 1$. If for $t \geq 0$ and $x \in \R^d$ we set 
\[
u_\t (t,x) = u \left(\frac t \t, \frac x \t \right),
\]
then the function $u_\t$ oscillates at frequency 1 and is solution of 
\[
\wp \left( \frac x \t \right) \partial_t^2 u_\t - \divg \Gp \left( \frac x \t \right) \nabla u_\t + \frac 1 \t \bp \left( \frac x \t \right) \partial_t u_\t = 0.
\]
This suggests that the first term should not play any role when $\t \to 0$. Moreover, at the limit the wave should only see the mean value of the highly oscillating damping $\bp\big( \frac x \t \big)$. We set 
\begin{equation} \label{def-bm}
\bm = \int_{\T} \wp (y) \ap(y) \, dy,
\end{equation}
where 
\[
\T = \left( -\frac 12, \frac 12 \right]^d.
\]
Similarly, for the second term, we consider the effective operator which describes the asymptotic behavior of the operator $- \divg \Gp \big( \frac x \t \big) \nabla$ at the limit $\t \to 0$. This is given by the periodic homogenization theory (see for instance \cite{BensoussanLionsPapanicolaou,allaire,tartar}). Let $\Gm$ be the H-limit of $\Gp(\frac x \t \big)$ when $\t$ goes to 0. This means that if $v_\t,v \in H^1(\R^d)$ and $f \in H\inv(\R^d)$ are such that
\[
-\divg \Gp\left(\frac x \t \right) \nabla v_\t = f \quad \text{and} \quad - \divg \Gm \nabla v = f,
\]
then, as $\t$ goes to 0,
\[
v_\t \rightharpoonup v \quad  \text{in $H^1(\R^d)$,} \quad \text{and} \quad   \Gp\left(\frac x \t \right) \nabla v_\t \rightharpoonup \Gm \nabla v \quad \text{in $L^2(\R^d)$.}
\]
In general, the matrix $\Gm$ is not the mean value of $\Gp$. If for $\x \in \R^d$ we denote by $\p_\x$ the $\Z^d$-periodic solutions of
\begin{equation} \label{def-psi}
-\divg  \Gp(x) (\x + \nabla \p_\x) = 0
\end{equation}
($\p_\x$ is defined up to a constant), and if we denote by $W(x)$ the $\Z^d$-periodic matrix such that  
\begin{equation} \label{def-J}
W(x) \x = \x + \nabla \p_\x(x),
\end{equation}
then $\Gm$ is in fact the mean value of $W(x)\trsp \Gp(x) W(x)$:
\begin{equation} \label{def-G-star}
\innp{\Gm \x}{\x} = \int_\T \innp{\Gp(x) (\x + \nabla \p_\x(x))}{(\x + \nabla \p_\x(x))} \, dx.
\end{equation}
Notice that it is natural to introduce all these quantities from the homogenization point of view (see \cite{ConcaVa97, OrtegaZu00, OrivePaZu01, ConcaOrVa02} for closely related contexts), but our proofs will be purely spectral. We will see in Section \ref{sec-periodic} how $\bm$, $\Gm$ and the functions $\p_\x$ naturally appear in this context.\\

Let 
\[
\Pm := -\divg \Gm \nabla.
\]
We now compare the solution $\upp$ of the dissipative wave equation \eqref{wave-per} with the solution $\um$ on $\R_+ \times  \R^d$ to the heat equation
\begin{equation} \label{heat-Q}
\bm \partial_t \um + \Pm \um = 0
\end{equation}
with initial condition
\begin{equation} \label{heat-Q-CI}
\restr{\um}{t=0} = \frac {\wp} {\bm} (\ap u_0 + u_1).
\end{equation}

After a linear change of variables, the estimates of \cite{MallougRo} for the standard heat equation read as follows.

\begin{proposition} \label{prop-heat}
Let $s_1,s_2 \in \big[ 0, \frac d 2 \big]$ and $\k > 1$. Let $s \in [0,1]$. Then there exists $C \geq 0$ such that for all $t \geq 1$ we have 
\begin{eqnarray*}
\nr{\pppg x^{-\k s_1} e^{-\frac {t \Pm}\bm} \pppg x^{- \k s_2}}_{\mathscr{L}(L^2(\R^d))} &\leq& C \pppg t^{-\frac {s_1 + s_2} 2},\\
\nr{\pppg x^{-\k s_1} \partial_t e^{-\frac {t \Pm}\bm}\pppg x^{-\k s_2}}_{\mathscr{L}(L^2(\R^d))} &\leq& C \pppg t^{- 1 - \frac {s_1 + s_2} 2},\\
\nr{\pppg x^{-\k s_1 - s} \nabla e^{-\frac {t \Pm}\bm} \pppg x^{- \k s_2 - s}}_{\mathscr{L}(L^2(\R^d))} &\leq& C \pppg t^{- \frac {1 +s} 2 - \frac {s_1 + s_2} 2}.
\end{eqnarray*}
\end{proposition}

Here and everywhere below, we denote by $\mathscr L (\Kc_1,\Kc_2)$ the space of bounded operators from $\Kc_1$ to $\Kc_2$. We also write $\mathscr{L}(\Kc_1)$ for $\mathscr{L}(\Kc_1,\Kc_1)$.\\

The main result of this paper is the following. We prove that the difference between the solution $\upp$ of \eqref{wave-per} and the solution $\um$ of \eqref{heat-Q}-\eqref{heat-Q-CI} decays faster that $\um$ (except for the gradient if $s = 1$, in which case we have the same estimate).

\begin{theorem}[Comparison with the heat equation] \label{th-heat}
Assume that the damping condition \eqref{hyp-GCC} holds. 
Let $s_1,s_2 \in \big[0,\frac d 2]$ and $\k > 1$. Then there exists $C \geq 0$ such that for $t \geq 0$ and $U_0 = (u_0,u_1) \in \HH^{\k s_2}$ we have 
\begin{eqnarray*}
\nr{\upp(t) - \um(t)}_{L^{2,-\k s_1}} &\leq& C \pppg t^{- \frac 12 - \frac {s_1 + s_2} 2} \nr{U_0}_{\HH^{\k s_2}},\\
\nr{\partial_t \big(\upp(t) - \um(t) \big)}_{L^{2,-\k s_1}} &\leq& C \pppg t^{- \frac 32 - \frac {s_1 + s_2} 2} \nr{U_0}_{\HH^{\k s_2}},\\
\nr{\nabla \upp(t) - W \nabla \um(t)}_{L^{2,-\k s_1}} &\leq& C \pppg t^{- 1 - \frac {s_1 + s_2} 2} \nr{U_0}_{\HH^{\k s_2}},
\end{eqnarray*}
where $\upp(t)$ and $\um$ are the solutions of \eqref{wave-per} and \eqref{heat-Q}-\eqref{heat-Q-CI}, respectively, and $W(x)$ is defined by \eqref{def-J}. Moreover $W(x)$ is bounded.
\end{theorem}

Here we compare the solution $\upp$ of the damped wave equation \eqref{wave-per} (depending on the metric $\Gp(x)$) with the solution $\um$ of a heat equation with the constant (homogenized) metric $\Gm$. We can also say that, at the first order, $u$ behaves like a solution of the heat equation with the metric $\Gp(x)$. Indeed, it is known that the solution of the heat equation with the periodic metric $\Gp(x)$ behaves itself at the first order like the solution of the heat equation with $\Gm$. See \cite{OrtegaZu00}.\\

We notice that the gradient of $\upp$ does not exactly behave like that of $\um$. We have to use the corrector matrix $W(x)$, but it is bounded, so it does not alter the estimate of $\nabla \um$.\\

With Proposition \ref{prop-heat}, Theorem \ref{th-heat} implies Theorem \ref{th-energy-decay-periodic}. More precisely, it confirms the energy decay estimates, it proves that they are sharp, and it shows that, as for the heat equation, we would not get better results by taking stronger (for instance, compactly supported) weights.
Thus, for compactly supported weights, we obtain the following estimates. For $R > 0$ there exists $C_R$ such that for $U_0$ supported in the ball $B(R)$ and $t \geq 0$ we have 
\begin{equation}
\nr{u(t)}_{L^2(B(R))} \leq C_R \pppg t^{-\frac d 2} \nr{U_0}_\HH
\end{equation}
and 
\begin{equation}
\nr{\partial_t u(t)}_{L^2(B(R))} + \nr{\nabla u(t)}_{L^2(B(R))} \leq C_R \pppg t^{-\frac d 2 -1} \nr{U_0}_\HH.
\end{equation}

The comparison between the damped wave equation and the corresponding heat equation with a periodic metric has already been analysed in \cite{OrivePaZu01}. Theorem \ref{th-heat} improves the result in different directions.

The main improvements concern the absorption index. First, it is not necessarily constant. This is an important difference for the spectral analysis of the operator corresponding to the wave equation, since in this case we do not necessarily have a Riesz basis. Moreover, this absorption index is allowed to vanish, which also makes some arguments used in \cite{OrivePaZu01} unavailable.

On the other hand, the main result of \cite{OrivePaZu01} provides an asymptotic developpement for localized initial data. More precisely, $(u_0,u_1) \in L^2(\R^d) \times H\inv(\R^d)$ belongs to some weighted $L^1$ space, and the more decay we have at infinity, the more precise the developpement is. Here we give estimates which are uniform in the energy of the initial data (however we still get better results for more localized initial data, and the dual remark is that the rate of decay will be better for the localized energy, even if the wave is dissipated at infinity).

However, compared to \cite{OrivePaZu01}, we give a less precise developpement. We only give the leading term, given by the solution $\um$ of \eqref{heat-Q}-\eqref{heat-Q-CI}. However, it may happen that $\ap u_0 + u_1 = 0$ (then $\um = 0$) or that its Fourier transform vanishes near 0 (then $\um$ decays exponentially). In these cases, we could get better estimates for the damped wave $u$ in Theorem \ref{th-energy-decay-periodic}.

In fact, we could continue the developpement for the purely periodic setting, but not for the general setting which we consider in this paper. Indeed, we allow a perturbation of all the periodic coefficients by asymptotically vanishing terms, which would invalidate the developpement. However, we will see that this does not alter the main term, so the estimates of Theorem \ref{th-energy-decay-periodic} remain valid. This is described in the following paragraph.

\subsection{Perturbation of the periodic setting}

In Theorems \ref{th-energy-decay-periodic} and \ref{th-heat} we have considered a purely periodic problem. Now we can state the generalizations of these results for the perturbed setting. 

\begin{theorem}[Perturbation of the periodic wave] \label{th-diff-decay}
Assume that the damping condition \eqref{hyp-GCC} holds. 
Let $\k >1$ and $s_1, s_2 ,\y \geq 0$ be such that 
\begin{equation} \label{hyp-s12-bis}
\max(s_1, s_2) + \y < \min\left(\frac d 2, \rho_G, \rho_a + 1 \right).
\end{equation}
Then there exists $C \geq 0$ such that for $U_0 = (u_0,u_1) \in \HH^{\k s_2}$ and $t \geq 0$ we have 
\begin{eqnarray*}
\nr{u(t) - \upp(t)}_{L^{2,-\k s_1}} &\leq& C \pppg t^{- \frac {s_1 + s_2} {2} - \frac \y 2 } \nr{U_0}_{\HH^{\k s_2}},\\
\nr{\partial_t \big(u(t) - \upp(t) \big)}_{L^{2,-\k s_1}} &\leq& C \pppg t^{-1 - \frac {s_1 + s_2} {2}- \frac \y 2} \nr{U_0}_{\HH^{\k s_2}},\\
\nr{\nabla \big(u(t) - \upp(t) \big)}_{L^{2,-\k s_1}} &\leq& C \pppg t^{ - \frac 1 2 - \frac {s_1 + s_2} {2}- \frac \y 2 } \nr{U_0}_{\HH^{\k s_2}},
\end{eqnarray*}
where $u(t)$ and $\upp(t)$ are the solutions of \eqref{wave} and \eqref{wave-per}, respectively.
\end{theorem}

With Theorems \ref{th-energy-decay-periodic} and \ref{th-diff-decay} we deduce the following estimates in the general setting:

\begin{corollary}[Energy estimates in the general setting] \label{cor-energy-decay-general}
Assume that the damping condition \eqref{hyp-GCC} holds. 
Let $\k > 1$, $s_1, s_2 \in \big[0,\frac d 2\big]$ and $s \in [0,1]$ be such that 
\[
\max(s_1,s_2) + s < \min\left(\frac d 2, \rho_G, \rho_a + 1 \right).
\]
Then there exists $C \geq 0$ such that for $U_0 = (u_0,u_1) \in \HH^{\k s_2+s}$ and $t \geq 0$ we have 
\begin{eqnarray*}
\nr{u(t)}_{L^{2,-\k s_1}} &\leq& C \pppg t^{- \frac {s_1 + s_2} 2} \nr{U_0}_{\HH^{\k s_2}},\\
\nr{\partial_t u(t)}_{L^{2,-\k s_1}} &\leq& C \pppg t^{-1 - \frac {s_1 + s_2} 2} \nr{U_0}_{\HH^{\k s_2}},\\
\nr{\nabla u(t)}_{L^{2,-\k s_1 -s}} &\leq& C \pppg t^{- \frac {1 + s} 2 - \frac {s_1 + s_2} 2} \nr{U_0}_{\HH^{\k s_2 + s}},
\end{eqnarray*}
where $u(t)$ is the solution of \eqref{wave}.
\end{corollary}

These estimates are the same as those of Theorem \ref{th-energy-decay-periodic}, even if there is a restriction in the choice of $s_1$ and $s_2$ when the perturbative coefficients $\Go$, $\ao$ and $\wo$ decay slowly at infinity. In particular, we recover exactly the same estimates as in the periodic case for the uniform global energy decay or if the perturbation is compactly supported.

\subsection{Organisation of the paper}

The paper is organized as follows. In Section \ref{sec-spectral} we introduce the wave operator in the energy space and its resolvent. In Section \ref{sec-high-freq} we discuss the contributions of high frequencies and explain how the problem reduces to the analysis of low frequencies. The main part of the paper is Section \ref{sec-periodic}, about the purely periodic case. We prove Theorem \ref{th-heat}, and Theorem \ref{th-energy-decay-periodic} will follow with Proposition \ref{prop-heat}. Finally, we consider the perturbed setting in Section \ref{sec-perturbed}.

\section{The Resolvent of the wave equation} \label{sec-spectral}

\newcommand{\spP}{\mathfrak s}

We will prove all the energy decay estimates from a spectral point of view. In this section we introduce the corresponding operators and give their basic spectral properties. Let 
\[
\C_+ = \set{z \in \C \st \Im(z) > 0}.
\]

We recall that an operator $T$ with domain $\Dom(T)$ on a Hilbert space $\Kc$ is said to be dissipative (respectively accretive) if 
\[
\forall \f \in \Dom(T), \quad \Im \innp{T\f}{\f}_\Kc \leq 0 \quad (\text{respectively,} \quad \Re \innp{T\f}\f_\Kc \geq 0).
\]
Then the operator $T$ is said to be maximal dissipative if $(T-z)$ is boundedly invertible for some (and therefore any) $z \in \C_+$. In this case we have, for all $z \in \C_+$,
\[
\nr{(T-z)\inv}_{\Lc(\Kc)} \leq \frac 1 {\Im(z)}.
\]
Moreover, if $T$ is also accretive, then $(T-z)$ is boundedly invertible when $\Re(z) < 0$ and we have 
\[
\nr{(T-z)\inv}_{\Lc(\Kc)} \leq \frac 1 {\abs{\Re(z)}}.
\]

We recall that $\PG$ and $b$ were defined after \eqref{wave-w}. If $z \in \C$ is such that the operator $\big( \PG - iz b(x) - z^2 w(x) \big) \in \mathscr{L}(H^2(\R^d),L^2(\R^d))$ has a bounded inverse, we set 
\[
R(z) = \big( \PG - iz b(x) - z^2 w(x) \big)\inv.
\]

\begin{proposition} \label{prop-Rz}
For $z \in \C_+$ the resolvent $R(z)$ is well defined and extends to a bounded operator from $H\inv(\R^d)$ to $H^1(\R^d)$. Moreover, we have $R(z)^* = R(-\bar z)$ and there exists $C \geq 0$ such that for $z \in \C_+$ we have 
\[
\nr{R(z)}_{\mathscr L (L^2(\R))} \leq \frac C {\Im (z) \abs z},
\]
\[
\nr{R(z)}_{\mathscr L (L^2(\R),H^1(\R))} + \nr{R(z)}_{\mathscr L (H\inv(\R),L^2(\R))} \leq \frac C {\Im(z)}
\]
and
\[
\nr{R(z)}_{\mathscr L (H\inv(\R),H^1(\R))} \leq \frac {C \abs z} {\Im(z)}.
\]
\end{proposition}

\begin{proof}
Let $z = \t + i \m \in \C$ with $\t \geq 0$ and $\m > 0$. We set 
\[
T(z) := \PG -izb - z^2 w = \big( \PG + \m b + (\m^2-\t^2) w \big) - i (\t b + 2\t \m w).
\]
Assume that $\m \leq 2\t$. Then $\tilde T(z) := T(z) + 2i\t \m w_{\min}$ is a dissipative and bounded perturbation of the selfadjoint operator $\PG$, so it is maximal dissipative. Thus $T(z) = \tilde T(z) - 2i\t\m w_{\min}$ is boundedly invertible and 
\[
\nr{T(z)\inv}_{\mathscr L (L^2(\R))} = \nr{\big(\tilde T(z) - 2i\t \m w_{\min} \big)\inv}_{\mathscr L (L^2(\R))} \leq \frac 1 {2\t \m w_{\min}}. 
\]
Now assume that $\m \geq 2\t$. Then $\tilde T(z) := T(z) - \frac {\m^2} 2 w_{\min}$ is a dissipative and accretive perturbation of the non-negative selfadjoint operator $\PG$, so $T(z) = \tilde T(z) + \frac {\m^2} 2 w_{\min}$ is boundedly invertible and 
\[
\nr{T(z)\inv}_{\mathscr L (L^2(\R))} = \nr{\left(\tilde T(z) + \frac {\m^2} 2 w_{\min}\right)\inv}_{\mathscr L (L^2(\R))} \leq \frac 2 {\m^2 w_{\min}}.
\]
In any case we have 
\[
\nr{T(z)\inv}_{\mathscr L (L^2(\R))} \lesssim \frac 1 {\m \abs z}. 
\]
If $\t < 0$ we observe that $T(z) = T(-\bar z)^*$ to obtain the same results. It only remains to prove the last two estimates. For $z \in \C_+$ and $\vf \in \Sc$ we have 
\[
\nr{\nabla R(z) \vf}_{L^2(\R)}^2 \lesssim \innp{\PG R(z) \vf}{R(z) \vf}  \lesssim \innp{\vf}{R(z) \vf} + \abs z^2 \nr{R(z) \vf}^2 \lesssim \frac 1 {\m^2} \nr{\vf}^2.
\]
This gives the estimate of the first term in the second inequality. The estimate of the second term follows by duality. For the last estimate we write
\begin{align*}
\nr{\nabla R(z) \nabla \vf}^2
& \lesssim \innp{\PG R(z) \nabla \vf}{R(z) \nabla \vf} \lesssim \innp{\nabla \vf}{R(z) \nabla \vf} + \abs z^2 \nr{R(z) \nabla \vf}^2\\
& \lesssim \nr{\nabla R(z) \nabla \vf} \nr \vf + \frac {\abs z^2} {\m^2} \nr \vf^2,
\end{align*}
and the conclusion follows.
\end{proof}

We consider on $\HH$ the operator 
\begin{equation} \label{def-Ac}
 \Ac = \begin{pmatrix} 0 & w\inv \\ \PG & -i a \end{pmatrix}
\end{equation}
with domain 
\begin{equation} \label{dom-A}
\Dom(\Ac) = H^2(\R^d) \times H^1(\R^d).
\end{equation}

Let $F = (u_0,i w u_1) \in \Dom(\Ac)$. Then $u$ is a solution to the problem \eqref{wave-w} if and only if $U = (u,i w \partial_t u)$ is a solution to
\begin{equation} \label{wave-A}
\begin{cases}
(\partial_t  + i \Ac  ) U(t) = 0,\\
U(0) = F.
\end{cases}
\end{equation}

\begin{proposition} \label{prop-Ac-max-diss}
For $z \in \C_+$ the operator $(\Ac-z)$ is boundedly invertible on $\HH$, and we have 
\[
(\Ac-z)\inv =
\begin{pmatrix} 
R(z) (ib + zw) &   R(z)\\
w +  w R(z) (izb + z^2 w)  &  zw R(z)
\end{pmatrix}.
\]
Moreover there exists $C \geq 0$ such that for all $z \in \C_+$ we have 
\[
\nr{(\Ac-z)\inv}_{\mathscr L (\HH)} \leq \frac C {\Im(z)}.
\]
\end{proposition}

\begin{proof}
Let $z = \t + i \m \in \C_+$, with $\t \in \R$ and $\m > 0$. For $F = (f,g) \in \HH$ we set 
\begin{align*}
\Rc_A(z) F
& =
\begin{pmatrix}
R(z) (ib + zw) f + R(z) g\\
wf + w R(z) (izb+z^2w) f + z w R(z) g
\end{pmatrix}\\
& =
\begin{pmatrix}
\frac 1 z R(z) \PG f - \frac 1 z f + R(z) g \\
w R(z) \PG f + z w R(z) g
\end{pmatrix}.
\end{align*}
With the first expression we see that $\Rc_A(z)$ is a bounded operator from $\HH$ to $\Dom(\Ac)$. By an explicit computation, we check that $\Rc_A(z)$ is an inverse for $(\Ac-z)$. Finally, with the second expression of $\Rc_A(z)$ and the estimates of Proposition \ref{prop-Rz}, we obtain 
\begin{align*}
\nr{\Rc_A(z)F}_{\HH} 
& \lesssim \frac 1 {\abs z} \nr{R(z)}_{\mathscr L(H\inv,H^1)} \nr{f}_{H^1} + \frac 1 {\abs z} \nr{f}_{H^1} + \nr{R(z)}_{\mathscr L (L^2,H^1)} \nr{g}_{L^2}\\
& + \nr{R(z)}_{\mathscr L (H\inv,L^2)} \nr{f}_{H^1} + \abs z  \nr{R(z)}_{\mathscr L (L^2)} \nr{g}_{L^2}\\
& \lesssim \frac {\nr F _{\HH}}{\m}.
\end{align*}
The proposition is proved.
\end{proof}

By the Hille-Yosida Theorem, we now deduce the following result about the propagator of $\Ac$. It ensures in particular that for $F \in \Dom(\Ac)$ the problem \eqref{wave-A} has a unique solution defined for all non-negative times.

\begin{proposition}
The operator $-i\Ac$ generates a semigroup on $\HH$. Moreover there exists $C \geq 0$ such that for all $t \geq 0$ we have 
\[
\nr{e^{-it\Ac}}_{\mathscr L (\HH)} \leq C.
\]
\end{proposition}

By Proposition \ref{prop-Ac-max-diss} we know that any $z \in \C_+$ belongs to the resolvent set of $\Ac$. As usual we are interested in the behavior of $(\Ac-z)\inv$ at the limit $\Im(z) \to 0$. In fact, with a strong decay, the spectrum is really under the real axis. Except for low frequencies\dots

\begin{theorem} \label{th-high-freq}
Any $\t \in \R \setminus \set 0$ belongs to the resolvent set of $\AAc$. Moreover there exists $C > 0$ such that for all $\t \in \R \setminus [-1,1]$ we have 
\begin{equation} \label{estim-high-freq}
\nr{(\Ac-\t)\inv}_{\mathscr{L}(\HH)} \leq C.
\end{equation}
\end{theorem}

For the proof of this result we refer to \cite{BurqJo} (notice that $w = 1$ in \cite{BurqJo}, but this does not play any role in this high-frequency analysis). 

The first statement about a fixed frequency holds under the general assumption that all the points in $\R^d$ are in some suitable sense uniformly close to the damping region (see Theorem 1.3 and Section 4 in \cite{BurqJo}). It is not difficult to check that this is always the case in our asymptotically periodic setting, even without the damping condition \eqref{hyp-GCC}.

Since the resolvent $(\Ac-\t)\inv$ is continuous on $\R \setminus \set 0$, it is clear that an estimate like \eqref{estim-high-freq} holds for $\t$ in a compact subset. However this resolvent may blow up when $\abs \t$ goes to $+\infty$. The fact that we have a uniform estimate even at the high-frequency limit relies on the damping condition \eqref{hyp-GCC} on classical trajectories (see Theorem 1.2 and Section 3 in \cite{BurqJo}). As explained in the introduction, we would have a weaker estimate with loss of regularity without this assumption. 

The proof of Theorem \ref{th-high-freq} relies on semiclassical analysis. This is why we need some regularity for the coefficients of the problem. Notice that \cite{BurqJo} requires uniform continuity for $a$. This is indeed the case here for our continuous and asymptotically periodic absorption index.

\begin{remark} \label{rem-energy-space}
All the estimates of the main theorems are given in $\HH$ or its weighted analogs. However, for the energy of a wave it would be more natural to work in the energy space $\EE$, defined as the Hilbert completion of $\Sc \times \Sc$ for the norm defined by
\[
\nr{(u,v)}_{\EE}^2 = \int_{\R^d}  G(x) \nabla u(x) \cdot \nabla \bar u (x)  dx + \int_{\R^d} \frac {\abs{v(x)}^2}{w(x)}  dx.
\]
We observe that $\EE$ is equal to the standard energy space $\dot H^1(\R^d) \times L^2(\R^d)$ with equivalent norm, and if $u$ is the solution of \eqref{wave} then its energy is exactly 
\[
E(t) = \nr{(u(t), w \partial_t u(t))}_\EE^2.
\]
Moreover we could check that the operator $\Ac$ would define on $\EE$ a maximal dissipative operator, so that $(e^{-it\Ac})_{t \geq 0}$ would be a contractions semigroup on $\EE$.

Working in $\EE$ instead of $\HH$ means that we are not interested in the size of the solution $u$ itself but only in the size of its first derivatives. And the estimates should not depend on $u_0$ but only on $\nabla u_0$ (see \cite{royer-dld-energy-space} for a discussion on this question). However for the heat equation it is natural to take into account the size of $u_0$. Thus, since our wave behaves like a solution of the heat equation, it is relevant to give all the estimates in $\HH$ instead of $\EE$.
\end{remark}

\section{Reduction to a low frequency analysis} \label{sec-high-freq}

In this section we show how we can use the resolvent estimate of Theorem \ref{th-high-freq} to reduce the time decay properties of Theorems \ref{th-energy-decay-periodic} and \ref{th-diff-decay} to the contributions of low frequencies. By density, it is enough to consider initial data in $\Sc \times \Sc$.\\

Let $\vf \in C^\infty(\R,[0,1])$ be equal to 0 on $(-\infty,1]$ and equal to 1 on $[2,+\infty)$. For $\e \in (0,1]$ and $t \in \R$ we set $\vf_\e(t) := \vf \big( \frac t \e \big)$, and then
\begin{equation} \label{def-Ueps}
U_\e(t) := \vf_\e(t) e^{-it\Ac}.
\end{equation}
Let $F \in \Sc \times \Sc$ and $\m \in (0,1]$. For $\t \in \R$ we have
\[
\int_\R e^{it\t} e^{-t\m} U_\e(t) F \, dt = -i \big(\AAc-(\t+i\m)\big)\inv F_\e(\t+i\m),
\]
where for $z \in \C$ we have set 
\begin{equation} \label{def-Fz}
F_\e(z) = \int_\e^{2\e} \vf_\e'(t) e^{-it(\Ac-z)} F \, dt.
\end{equation}
By Theorem \ref{th-high-freq}, the map $\t \mapsto \big(\AAc-(\t+i\m)\big)\inv F_\e(\t+i\m)$ belongs to $\Sc$. Then the Fourier inversion formula yields, for all $t \in \R$,
\[
e^{-t\m} U_\e(t) F = \frac {1} {2i\pi} \int_{\R} e^{-it\t} \big( \AAc-(\t+i\m) \big)\inv F_\e(\t+i\m) \, d\t,
\]
or
\begin{equation} \label{eq-Ueps}
U_\e(t) F = \frac {1} {2i\pi} \int_{\Im(z) = \m} e^{-itz} (\AAc-z)\inv F_\e(z) \, dz.
\end{equation}
Let $C > 0$ be given by Theorem \ref{th-high-freq} and $\g \in \big(0, \frac 1 {2C}\big)$. Then the resolvent $(\Ac-z)\inv$ is well defined if $\abs{\Re(z)} \geq 1$ and $\Im(z) \geq - \g$. We consider $\th_\m \in C^\infty(\R,\R)$ such that $\th_\m(s) = \m$ if $\abs{s} \leq 1$, $-\g\leq \th_\m(s) \leq \m$ if $\abs s \in [1,2]$ and $\th_\m(s) = -\g$ if $\abs s \geq 2$. Then we set (see Figure \ref{fig-Gamma})
\begin{equation} \label{def-Gamma}
\G_\m := \set{\t + i \th_\m(\t) , \t \in \R}.
\end{equation}
\begin{figure}[h]
\begin{center}
\resizebox{0.8\textwidth}{!}{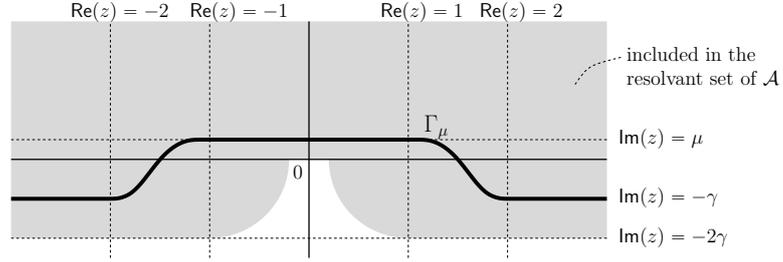}\\[2mm]
\end{center}
\caption{The curve $\Gamma_\m$.} \label{fig-Gamma}
\end{figure}

Since the integrand in \eqref{eq-Ueps} is holomorphic and decays rapidly at infinity we can write 
\[
U_\e(t) F = \frac 1 {2i\pi} \int_{\G_\m} e^{-itz} (\AAc-z)\inv F_\e(z) \, dz.
\]
Notice that, by holomorphy of the integrand, the right-hand side does not depend on $\m \in (0,1]$. Then we separate the contributions of low and high frequencies. For this we consider $\h \in C_0^\infty(\R,[0,1])$ supported in (-3,3) and equal to 1 on a neighborhood of [-2,2]. For $F \in \Sc \times \Sc$ we set
\[
\Ulow(t) F = \frac {1}{2i\pi}\int_{\Gamma_\m} \chi(\Re(z)) e^{-itz}(\Ac-z)^{-1} F_\e(z)~dz
\]
and 
\[
\Uhigh(t) F = \frac {1}{2i\pi}\int_{\Gamma_\m} (1-\chi)(\Re(z))
e^{-itz}(\Ac- z)^{-1} F_\e(z)~dz.
\]
Again, these quantities do not depend on $\m$ (this is clear for $\Uhigh(t) F$, for $\Ulow(t) F$ it follows from the holomorphy of the integrand in the region where $\abs{\Re(z)} \leq 2$). We begin with the contribution of high frequencies:

\begin{proposition} \label{prop-high-freq}
There exists $C \geq 0$ such that for $F \in \Sc \times \Sc$, $\m \in (0,1]$, $\e \in (0,1]$ and $t \geq 0$ we have 
\[
\nr{\Uhigh(t) F}_{\HH} \leq \frac {C e^{-\frac {\g t} 2}} {\sqrt \e}  \nr{F}_\HH.
\]
\end{proposition}

\begin{proof}
Let $F \in \Sc \times \Sc$. For $\e \in (0,1]$ and $t \in \R$ we set
\begin{equation} \label{def-Ieps}
I_\e(t) := e^{\g t} \Uhigh(t) F.
\end{equation}
We have 
\[
I_\e(t) = \frac 1 {2i\pi} \int_{\R} (1-\h)(\t) e^{-it\t} (\AAc-(\t-i\g))\inv F_\e(\t - i\g) \, d\t.
\]
By the Plancherel equality (twice) and Theorem \ref{th-high-freq} we have 
\begin{eqnarray} 
\nonumber
\lefteqn{\int_{\R} \nr{I_\e(t)}_{\HH}^2 \, dt
\lesssim \int_{\R} \nr{(1-\h)(\t) (\AAc-(\t-i\g))\inv F_\e(\t - i\g)}_{\HH}^2 \, d\t}\\
\nonumber
&& \lesssim \int_{\R} \nr{F_\e(\t - i\g)}_{\HH}^2 \, d\t
 \lesssim \int_{\R} \nr{\vf_\e'(t) e^{-it\AAc} e^{\g t} F}_\HH^2 \, dt
 \lesssim \frac {\nr{F}_\HH^2} {\e^2} \int_{\e}^{2\e} e^{2 \g t}  \, dt\\
\label{double-Plancherel}
&& \lesssim \frac {\nr{F}_{\HH}^2} \e.
\end{eqnarray}
Let $t_0,t \in \R$ with $t_0 < t$. For $s \in [t_0,t]$ we have 
\[
 \frac d {ds} \left(e^{-i(t-s)\Ac} I_\e(s) \right) = \g e^{-i(t-s)\Ac} I_\e(s) + \frac 1 {2i\pi} \int_{\R} (1-\h)(\t) e^{-it\t} e^{-i(t-s)\AAc} F_\e(\t - i\g) \, d\t,
\]
so as above we can check that 
\[
\int_{t_0}^{t} \nr{\frac d {ds} \left(e^{-i(t-s)\Ac} I_\e(s) \right)}_{\HH}^2 \, ds \lesssim \frac {\nr{F}_\HH^2} \e.
\]
Then, by the Cauchy-Schwarz inequality,
\begin{align*}
\nr{I_\e(t)}_{\HH}
& \leq \nr{e^{-i(t-t_0)\AAc} I(t_0)}_\HH + \int_{t_0}^{t} \nr{\frac d {ds} \left(e^{-i(t-s)\AAc} I_\e(s) \right)}_{\HH} \, dt\\
& \leq  \nr{I_\e(t_0)}_\HH + \sqrt {t-t_0} \frac {\nr{F}_{\HH}}{\sqrt \e}.
\end{align*}
By \eqref{double-Plancherel} we have
\[
\inf_{t_0 \in [0,1]} \nr{I_\e(t_0)}_\HH \lesssim \frac {\nr{F}_\HH} {\sqrt \e},
\]
so for $t \geq 1$
\[
\nr{I_\e(t)}_\HH \lesssim \sqrt {\frac t \e}  {\nr{F}_\HH}.
\]
With \eqref{def-Ieps}, this concludes the proof.
\end{proof}

We now turn to the contribution of low frequencies. The smooth cut-off $\vf_\e$ introduced in \eqref{def-Ueps} was useful to analyse the contribution of high frequencies (if $U_\e$ is smooth then $F_\e(z)$ is small at infinity). For low frequencies we could also estimate $U_\e$ for some fixed $\e$, but in order to obtain the sharp result of Theorem \ref{th-heat} we have to work with the initial data $F$ and not its perturbed version $F_\e$. In the following lemma we let $\e$ go to 0. Since $-\vf_\e'$ somehow converges to the Dirac mass at $t = 0$, we obtain that we can replace $F_\e$ by $F$ in the expression of $U_\e^\low(t)$. We set  
\begin{equation} \label{def-Ilow}
\Ilow(t)
:= \frac {1}{2i\pi}\int_{\Gamma_\m} \chi(\Re(z)) e^{-itz}(\Ac-z)^{-1} \,dz.
\end{equation}
As above, this does not depend on $\m \in (0,1]$.

\begin{proposition}\label{prop-low-freq}
There exists $C \geq 0$ such that for $F \in \Sc \times \Sc$, $\e \in (0,1]$ and $t \geq 0$ we have 
\[
\nr{\Ulow(t) F - \Ilow(t) F}_\HH \leq C \e \nr{F}_\HH.
\]
\end{proposition}

\begin{proof}
Let $\e \in (0,1]$. For $z \in \C_+$ we have 
\[
(\AAc-z)\inv F = i \int_0^{+\infty} e^{-is(\AAc-z)} F\, ds.
\]
On the other hand
\[
(\AAc-z)\inv F_\e(z) 
= \int_0^{+\infty} \vf_\e'(s) e^{-is(\AAc-z)} (\AAc-z)\inv F \, ds 
= i \int_0^{+\infty} \vf_\e(s) e^{-is(\AAc-z)} F \, ds,
\]
so
\[
(\AAc-z)\inv \big( F-F_\e(z) \big) = i \int_0^{2\e} \big(1-\vf_\e(s)\big) e^{-is(\AAc-z)} F \, ds.
\]
Let $\m \in (0,1]$. This equality between holomorphic functions on $\C_+$ can be extended to any $z \in \G_\m$. Moreover, since we only integrate over a compact subset of $\G_\m$ we can write 
\[
\nr{\Ulow(t) F - \Ilow(t) F}_\HH \lesssim \sup_{\substack {z \in \G_\m \\ \abs{\Re(z)} \leq 3}} \nr{e^{-itz} (\AAc-z)\inv \big( F-F_\e(z) \big)}_\HH \lesssim \e e^{\m t} \nr{F}_\HH.
\]
Since the left-hand side does not depend on $\m \in (0,1]$, we can let $\m$ go to 0, which concludes the proof.
\end{proof}

By Proposition \ref{prop-high-freq} and Lemma \ref{prop-low-freq} applied with $\e = e^{-\frac {\g t} 4}$, we finally obtain the following result:

\begin{proposition} \label{prop-eAAc-Ilow}
There exists $C \geq 0$ such that for $t \geq 0$ and $F \in \HH$ we have
\[
\nr{e^{-it\AAc} F - \Ilow(t) F}_\HH \leq C e^{-\frac {\g t} 4} \nr F_\HH.
\]
\end{proposition}

The rest of the paper is devoted to the analysis of $\Ilow(t) F$.

\section{Low frequency analysis in the periodic setting} \label{sec-periodic}

Let 
\begin{equation} \label{def-Ilowp}
\Ilowp(t)
:= \frac {1}{2i\pi}\int_{\Gamma_\m} \chi(\Re(z)) e^{-itz}(\AAp-z)^{-1} \, dz.
\end{equation}
This coincides with $\Ilow(t)$ (see \eqref{def-Ilow}) in the particular case of a purely periodic setting. In this case the result of Proposition \ref{prop-eAAc-Ilow} gives 
\begin{equation} \label{estim-high-freq-periodic}
\nr{e^{-it\AAp} F - \Ilowp(t) F}_\HH \leq C e^{-\frac {\g t} 4} \nr F_\HH.
\end{equation}
In this section we analyse $\Ilowp(f)$. With \eqref{estim-high-freq-periodic}, this will prove Theorem \ref{th-heat}, and hence Theorem \ref{th-energy-decay-periodic}.

\subsection{Floquet-Bloch decomposition of the periodic problem}

If $G(x) = \Gp(x)$, $a(x) = \ap(x)$ and $w(x) = \wp(x)$, then the medium in which our wave propagates is exactly $\Z^d$-periodic. However, the initial data and the solution itself are not periodic, so we cannot see our problem as a problem on the torus. We will use the Floquet-Bloch decomposition to write a function in $L^2(\R^d)$ as an integral of $\Z^d$-periodic contributions.\\

We denote by $L^2_\per$ the space of $L^2_\loc$ and $\Z^d$-periodic functions on $\R^d$. It is endowed with the natural norm defined by
\[
\nr u_{L^2_\per}^2 := \int_\T \abs{u(x)}^2 \, dx.
\]
Then we set $\LL_\per = L^2_\per \times L^2_\per$. For $k \in \N$ we also define $H^k_\per$ as the space of $\Z^d$-periodic and $H^k_\loc$ functions, endowed with the obvious norm.\\

The Floquet-Bloch decomposition is standard in this kind of context. We begin this section by recording the definitions and properties which we are going to use in this paper. 
For $u \in \Sc$, $\s \in \R^d$ and $x \in \R^d$ we set
\begin{equation} \label{def-uper-s}
u_\per^\s (x) = \sum_{n \in \Z^d} u(x + n) e^{-i(x+n)\cdot \s}.
\end{equation}
For all $\s \in \R^d$ the function $u_\per^\s$ belongs to $L^2_\per$.

\begin{proposition} \label{prop-Floquet-Bloch}
Let $u,v \in \Sc$.
\begin{enumerate}[(i)]
\item For $x \in \R^d$ we have 
\[
u(x) = \frac 1 {(2\pi)^d} \int_{\s \in 2\pi\T} e^{ix\cdot \s} u_\per^\s(x) \, d\s.
\]
\item For $\p \in L^2_\per$ and $\s \in \R^d$ we have 
\[
\innp{u_\per^\s}{\p}_{L^2_\per} = \int_{x \in \R^d} e^{-ix\cdot \s} u(x) \bar {\p(x)} \, dx.
\]
\item We have 
\[
\nr{u}_{L^2(\R^d)}^2 = \frac 1 {(2\pi)^d} \int_{\s \in 2\pi\T} \nr{u_\per^\s}_{L^2_\per}^2 \, d\s
\]
or, more generally,
\[
\innp{u}{v}_{L^2(\R^d)} = \frac 1 {(2\pi)^d} \int_{\s \in 2\pi\T} \innp{u_\per^\s}{v_\per^\s}_{L^2_\per}^2 \, d\s.
\]
\end{enumerate}
\end{proposition}

\begin{proof}
For the first statement we only have to write 
\[
\int_{\s \in 2\pi\T} e^{ix\cdot \s} u_\per^\s(x) \, d\s = \sum_{n \in \Z^d} u(x+n) \int_{\s \in 2\pi\T} e^{-in\cdot \s}\, d\s = (2\pi)^d u(x).
\]
The second property follows from 
\begin{align*}
\innp{u_\per^\s}{\p}_{L^2_\per}
& = \int_{y \in \T} \sum_{n \in \Z^d} u(y + n) e^{-i(y+n)\cdot \s} \bar {\p(y)} \, dy\\
& = \sum_{n \in \Z^d} \int_{y \in \T} u(y + n) e^{-i(y+n)\cdot \s} \bar {\p(y + n)} \, dy\\
& = \int_{x \in \R^d} u(x) e^{-ix\cdot \s} \bar {\p(x)} \, dx. 
\end{align*}
In particular
\begin{align*}
\int_{\s \in 2\pi\T} \innp{u_\per^\s}{v_\per^\s}_{L^2_\per}^2 \, d\s
& = \int_{\s \in 2\pi\T} \int_{x \in \R^d} e^{-ix \cdot \s} u(x) \sum_{n \in \Z^d} \bar{v(x+n)} e^{i (x+n) \cdot \s} \, dx \, d\s\\
& = \int_{x \in \R^d} u(x) \sum_{n \in \Z^d} \bar{v(x+n)} \int_{\s \in 2\pi\T} e^{i n \cdot \s} \, d\s \, dx\\
& = (2\pi)^d \int_{x \in \R^d} u(x) \bar{v(x)} \, dx.
\end{align*}
The proof is complete.
\end{proof}

If $u \in L^1(\R^d)$ and $\p \in L^2_\per \cap L^\infty(\R^d)$ then by Proposition \ref{prop-Floquet-Bloch} we have for all $\s \in \R^d$
\[
\abs{\innp{u_\per^\s}{\p}_{L^2_\per}} \leq \nr{u}_{L^1(\R^d)} \nr{\p}_{L^\infty(\R^d)}.
\]
If $\p$ is not assumed to be in $L^\infty$ but $u \in L^{2,\d}$ for some $\d > \frac d 2$ (then $L^{2,\d} \subset L^1$) we have a similar estimate. More generally, we have the following result.

\begin{corollary} \label{cor-Floquet-Bloch}
Let $\k > 1$. Let $s \in \big[0,\frac d 2]$ and $p = \frac {2d}{d-2s} \in [2,+\infty]$. Then there exists $C \geq 0$ such that for $u \in \Sc$ and $\p_\s \in L^\infty_\s(2\pi \T, L^2_\per)$ we have 
\[
\nr{\innp{u_\per^\s}{\p_\s}_{L^2_\per}}_{L^p_\s(2\pi\T)} \leq C \nr{u}_{L^{2,\k s}(\R^d)} \nr{\p_\s}_{L^\infty_\s(2\pi \T,L^2_\per)}.
\]
\end{corollary}

\begin{proof}
The case $s = 0$, $p = 2$, simply follows from the Cauchy-Schwarz inequality and Proposition \ref{prop-Floquet-Bloch}. For the case $s = \frac d 2$ and $p = \infty$ we use again Proposition \ref{prop-Floquet-Bloch} and the Cauchy-Schwarz inequality to write 
\begin{align*}
\abs{\innp{u_\per^\s}{\p_\s}_{L^2_\per}}
& \leq \int_{\R^d} \pppg x^{\frac {\k d}2} \abs{u(x)} \pppg x^{-\frac {\k d}2} \abs{\p_\s(x)} \, dx\\
& \leq \nr{u}_{L^{2,\frac {\k d}2}} \left(\int_{\T} \abs{\p_\s(y)}^2 \sum_{n \in \Z^d} \pppg {y+n}^{- \k d}  \, dy \right)^{\frac 12}\\
& \lesssim \nr{u}_{L^{2,\frac {\k d}2}} \nr{\p_\s}_{L^2_\per}.
\end{align*}
The general case follows by interpolation (we recall that for $\th \in [0,1]$ we have $L^{2,\th \k d /2} = (L^2,L^{2,\k d/2})_{[\th]}$ and $(L^2,L^\infty)_{[\th]} = L^p$ with $1/p = (1-\th)/2$)).
\end{proof}

\begin{remark}
Notice that it is usual (see for instance Theorem 4.3.1 in \cite{BensoussanLionsPapanicolaou}) to decompose directly $u_\per^\s$ with respect to the basis of $L^2_\per$ given by the eigenfunctions for the (selfadjoint) periodic problem under study (the Bloch waves). This strategy is used in \cite{OrivePaZu01} for the wave equation with constant damping. In this case, the eigenfunctions of the wave operator are related to those of the Laplacian operator, which form a Hilbert basis. The same strategy cannot be used here with a non-constant absorption index.
\end{remark}

Let
\begin{equation} \label{def-AAp}
\AAp = 
\begin{pmatrix}
0 & \wp\inv \\
\Pp & -i\ap
\end{pmatrix}
\end{equation}
(notice that all the results of Section \ref{sec-spectral} hold in particular when $G = \Gp$, $a = \ap$ and $w = \wp$).
For $u \in \Sc$ and $x \in \R^d$ we can write
\begin{equation} \label{bloch-P}
\Pp u (x) = \frac 1 {(2\pi)^d} \int_{\s \in 2\pi\T} \Pp e^{ix\cdot \s} u_\per^\s(x) \, d\s = \frac 1 {(2\pi)^d} \int_{\s \in 2\pi\T} e^{ix\cdot \s} \Psp u_\per^\s(x) \, d\s,
\end{equation}
where for $\s \in \R^d$ we have set
\begin{equation*}
\Psp = e^{-ix\cdot \s} \Pp e^{ix\cdot \s} = - (\divg + i \s\trsp)  \Gp(x) (\nabla + i\s).
\end{equation*}
Now let $U = (u,v) \in \Sc \times \Sc$. For $\s \in \R^d$ and $x \in \R^d$ we set $U_\per^\s(x) = \big(u_\per^\s (x), v_\per^\s (x)\big)$. Then we write 
\begin{equation} \label{bloch-Ac}
\AAp U = \frac 1 {(2\pi)^d} \int_{\s \in 2\pi\T} e^{ix\cdot \s} \AAs U_\per^\s \, d\s,
\end{equation}
where 
\[
\AAs = 
\begin{pmatrix}
0 & \wp\inv\\
\Psp & -i\ap
\end{pmatrix}.
\]

The interest of the decomposition \eqref{bloch-Ac} of the operator $\AAp$ is that each $\AAs$ has a compact resolvent, hence its spectrum is given by a sequence of isolated eigenvalues of finite algebraic multiplicities:

\begin{proposition}
Let $\s \in \R^d$.
\begin{enumerate}[(i)]
\item Then $\AAs$ defines an operator on $H^1_\per \times L^2_\per$ with domain $H^2_\per \times H^1_\per$. Moreover, it has a compact resolvent.
\item Let $z \in \C$. Then $(\AAs -z) \in \mathscr L (H^2_\per \times H^1_\per, H^1_\per \times L^2_\per)$ has a bounded inverse if and only if $(\Psp - iz\bp - z^2 \wp) \in \mathscr L (H^2_\per, L^2_\per)$ has a bounded inverse, which we denote by $R_\s(z)$, and in this case we have 
\begin{equation} \label{expr-res-AAs}
(\AAs -z)\inv = 
\begin{pmatrix}
R_\s(z) (i\bp + z\wp) & R_\s(z)\\
\wp + R_\s(z) (iz\bp + z^2\wp) & zR_\s(z)
\end{pmatrix}.
\end{equation}
In particular, $(\AAs-z)\inv$ extends to a bounded operator from $L^2_\per \times H\inv_\per$ to $H^1_\per \times L^2_\per$. 
\item Any $z \in \C_+$ belongs to the resolvent set of $\AAs$.
\end{enumerate}
\end{proposition}

\begin{proof}
\stepp The operator $\Psp$ is selfadjoint on $L^2_\per$ with domain $H^2_\per$. As in the proof of Proposition \ref{prop-Rz}, we can check that for $z \in \C_+$ the operator $(\Psp-iz\bp-z^2\wp)$ indeed has a bounded inverse, and that when $R_\s(z)$ is well defined in $\mathscr L(L^2_\per,H^2_\per)$ it extends to a bounded operator from $H\inv_\per$ to $H^1_\per$.

\stepp Let $z \in \C$. If $R_\s(z)$ is well defined, then we can check by direct computation that the right-hand side of \eqref{expr-res-AAs} defines a bounded inverse for $(\AAs-z)\inv$. Conversely, assume that $z$ belongs to the resolvent set of $\AAs$. Then for $g \in L^2_\per$ we set 
\[
U = \begin{pmatrix} u \\ v \end{pmatrix} = (\AAs-z)\inv \begin{pmatrix} 0 \\ g \end{pmatrix}
\]
and
\[
R_\s(z) g = u \in H^2_\per.
\]
This defines a bounded operator from $L^2_\per$ to $H^2_\per$. Moreover, we compute $(\AAs-z) U$ and get 
\[
(\Psp-iz\bp-z^2\wp) u = g,
\]
which proves that $R_\s(z)$ is an inverse for $(\Psp-iz\bp-z^2\wp)$.

\stepp Finally we observe that $H^2_\per \times H^1_\per$ is compactly embedded in $H^1_\per \times L^2_\per$, so $\AAs$ has a compact resolvent, and the proof is complete.
\end{proof}

For $F \in \Sc \times \Sc$ and $z \in \C_+$ we have 
\[
(\AAp-z)\inv F = \frac 1 {(2\pi)^d} \int_{\s \in 2\pi\T} e^{ix\cdot \s} (\AAs-z)\inv F_\per^\s \, d\s,
\]
where $(\AAs-z)\inv$ is as given by \eqref{expr-res-AAs}. The equality remains valid for any $z$ in the resolvent sets of $\AAp$ and $\AAs$ for all $\s \in 2\pi \T$.

\subsection{Reduction to the contributions of small \texorpdfstring{$\s$}{sigma} and of the first Bloch wave}

With the Floquet-Bloch decomposition we have somehow reduced the spectral analysis of $\AAp$ to an eigenvalue problem for the family of operators $\AAs$, $\s \in 2\pi\T$. Because of the non-selfadjointness of these operators, the corresponding sequences of eigenfunctions do not form an orthogonal basis (and, in fact, not even a Riesz basis), but we can show that the decay of $\Ilowp(t) F$ is only governed by the contribution of $\s$ close to 0 and of the ``first'' eigenvalue of the operator $\AAs$. This is the purpose of this paragraph.\\

We first observe that for $\s \in \R^d$, $\l \in \C$ and $U = (u,v) \in H^2_\per \times H^1_\per$ we have 
\begin{equation} \label{vp-AAs}
\AAs U = \l U \quad \eqv \quad 
\begin{cases}
\big(\Psp - i \l \bp - \l^2 \wp \big) u = 0,\\
v = \l \wp u.
\end{cases}
\end{equation}

\begin{proposition} \label{prop-eigenvalues}
The following assertions hold.
\begin{enumerate}[(i)]
\item If $\l \in \Sp(\AAs)$ for some $\s \in 2\pi \T$, then $\Im(\l) \leq 0$.
\item There exist $\rr > 0$, $\ggg_2 > 0$ and $\ggg_1 \in (0,\min(1,\ggg_2))$ such that for $\s \in B(\rr)$ the operator $\AAs$ has a unique eigenvalue $\ls$ with $\abs {\ls} \leq \g_1$ and all the other eigenvalues with real part in {$[-3,3]$} have an imaginary part smaller than $-\ggg_2$. Moreover the eigenvalue $\ls$ is algebraically simple.
\item There exists $\ggg_0 \in (0,\ggg_1)$ such that for $\s \in 2\pi\T \setminus B(\rr)$ and $\l \in \Sp(\AAs)$ with $\abs{\Re(\l)} \leq 3$ we have $\Im(\l) \leq - \ggg_0$.
\end{enumerate}
\end{proposition}

Without loss of generality we can assume that the constant $\g > 0$ used in the definition of $\G_\m$ (see \eqref{def-Gamma}) is smaller than $\ggg_0$.

\begin{proof}
\stepp Let $\s \in \bar{2\pi \T}$, $\l \in \Sp(\AAs)$ and let $U = (u,v) \in H^2_\per \times H^1_\per$ be a corresponding eigenvector. By \eqref{vp-AAs} we have 
\begin{equation} \label{innp-vp}
\innp{\big(\Psp - i \l \bp - \l^2 \wp \big) u}{u}_{L^2_\per} = 0.
\end{equation}
Taking the real and imaginary parts gives  
\begin{equation} \label{innp-vp-real}
\innp{\Psp u}{u} + \Im(\l) \innp{\bp u}{u} + \big( \Im(\l)^2 - \Re(\l)^2 \big) \innp{\wp u}{u} = 0
\end{equation}
and
\begin{equation} \label{innp-vp-im}
- \Re(\l) \innp{\bp u}{u} - 2 \Re(\l) \Im(\l) \innp{\wp u} u = 0.
\end{equation}
Assume that $\Re(\l) \neq 0$ and $\Im(\l) \geq 0$. By \eqref{innp-vp-im} we have $\bp u = 0$, which implies in particular that $\Psp u - \l^2 \wp u = 0$. Since $\bp$ is not identically zero, this also implies that $u$ vanishes on an open subset of $\R^d$. Thus $\tilde u : x \mapsto e^{ix\cdot \s} u(x)$ vanishes on an open subset of $\R^d$ and is a solution of $\Pp \tilde u - \l^2 \wp \tilde u = 0$. 
By unique continuation we have $\tilde u = 0$ and hence $u = 0$. Then $v = 0$ and $U = 0$, which gives a contradiction. If $\Re(\l) = 0$ and $\Im(\l) > 0$ then all the terms in \eqref{innp-vp-real} are non-negative. Again, we have $\bp u = 0$ and we get a contradiction. This proves the first statement and the fact that 0 is the only possible real eigenvalue.

\stepp Now assume that $\l = 0$, so that $\AAs U = 0$. By \eqref{vp-AAs} we have $v = 0$ and 
\[
\innp{\Gp(x) (\nabla + i\s) u}{(\nabla + i\s) u}_{L^2_\per} = \innp{\Psp u}{u}_{L^2_\per} = 0,
\]
so $(\nabla + i \s) u = 0$. Since $u$ is periodic and non-zero, this is only possible if $\s = 0$ and $u$ is constant. Conversely, if $u$ is constant we indeed have $U = (u,0) \in H^2_\per \times H^1_\per$ and $\AAs U = 0$. This proves that 0 is an eigenvalue of $\AAs$ if and only if $\s = 0$, and that 0 is a geometrically simple eigenvalue of $\AAso$.
Since $\AAso$ is not selfadjoint, it may have Jordan blocks, so we also have to prove that $\ker (\AAso^2) \subset \ker(\AAso)$. Let $U = (u,v) \in \Dom (\AAso^2)$ be such that $\AAso^2 U = 0$. Since $\AAso U \in \ker(\AAso)$ there exists $\a \in \C$ such that $\AAso U = (\a,0)$, which gives 
\[
\begin{cases}
\wp\inv v = \a,\\
\Pp u - i \ap v = 0.
\end{cases}
\]
Then, since $u$ is periodic, we have 
\[
0 = \int_{\T} \Pp u = i \a \int_{\T} \bp.
\]
This implies that $\a = 0$, and hence $U \in \ker(\AAso)$. Finally, 0 is an algebraically simple eigenvalue of $\AAso$.

\stepp 
The family of operators $(\AAs)_{\s \in \R^d}$ on $\LL_\per$ is analytic of type B in the sense of Kato (see \cite{kato}) with respect to each $\s_j$, $j \in \Ii 1 d$. 
Since 0 is a simple and isolated eigenvalue of $\AAso$, there exist $\rr > 0$ and $\ggg_1 > 0$ such that for $\s \in B(\rr)$ the operator $\AAs$ has a unique eigenvalue $\ls$ in the disk $D(0,\ggg_1)$ of $\C$. Moreover, this eigenvalue is algebraically simple. 
Let $\s \in \bar {B(\rr)}$. There exists $\ggg_{\s} > 0$ and a neighborhood $\Vc_\s$ of $\s$ such that if $s \in \Vc_\s$ and $\l \in \Sp(\AAp^{s}) \setminus \set{\l_{s}}$ with $\Re(\l) \in [-3,3]$ then $\Im(\l) \leq - \ggg_{\s}$. Since $\bar {B(\rr)}$ is compact, we can find $\s_1,\dots,\s_k \in \bar{B(\rr)}$ such that $\bar{B(\rr)} \subset \bigcup_{j=1}^k \Vc_{\s_j}$. Then we set $\ggg_2 = \min \set{\ggg_{\s_j}, 1 \leq j \leq k}$. Choosing $\rr$ and $\ggg_1$ smaller if necessary we have $\ggg_2 > \ggg_1$, which gives the second statement.

\stepp Using the same continuity and compactness argument we can check that there exists $\ggg_0 > 0$ such that for $\s \in \bar {2\pi\T \setminus B(\rr)}$ and $\l \in \Sp(\AAs)$ with $\abs{\Re(\l)} \leq 3$ we have $\Im(\l) \leq - \ggg_0$. This concludes the proof of the proposition.
\end{proof}

For $\s \in B(\rr)$ we set in $\mathscr{L}(H^1_\per \times L^2_\per)$
\[
\Ps = - \frac 1 {2i\pi} \int_{\abs \z = \ggg_1} (\AAs - \z)\inv \, d\z.
\]
It is known (see for instance \cite{kato}) that $\Ps$ is the projection on the line spanned by the eigenfunctions corresponding to the eigenvalue $\ls$ and along the subspace spanned by all the generalized eigenfunctions corresponding to all the other eigenvalues. In particular,
\[
\Ran(\Pso) = \set{ (\a,0) , \a \in \C}.
\]
Moreover it is a holomorphic function of $\s_j$ for all $j \in \Ii 1 d$ and maps $H^1_\per \times L^2_\per$ to $H^{k+1}_\per \times H^k$ for all $k \in \N$. It also extends to a bounded operator on $\LL_\per$.
We denote by $\Phi_0$ the constant function 
\[
\Phi_0 = \begin{pmatrix} 1 \\ 0 \end{pmatrix}.
\]
Choosing $\rr > 0$ smaller if necessary, we can assume that $\Ps \Phi_0 \neq 0$ for all $\s \in B(\rr)$. Then for $\s \in B(\rr)$ we set 
\[
\Phi_\s = \frac {\Ps \Phi_0}{\nr{\Ps \Phi_0}_{\LL_\per}}.
\]
Then $\nr{\Phi_\s}_{\LL_\per} = 1$ and $\AAs \Phi_\s = \ls \Phi_\s$ for all $\s \in B(r)$. By \eqref{vp-AAs}, there exists $\f_\s \in H^2_\per$ such that 
\begin{equation} \label{Phi-phi}
\Phi_\s = \begin{pmatrix} \f_\s \\ \ls \wp \f_\s \end{pmatrix}.
\end{equation}
Moreover $\f_0 \equiv 1$ and $\f_\s$ is a smooth function of $\s$.\\

In the following proposition we show that in $\Ilowp(t) F$ the important contribution is given by $\ls$ for $\s$ small. For $t \geq 0$ and $F \in \Sc \times \Sc$ we set 
\begin{equation} \label{def-tilde-Ilow}
\tIlow(t) F = \begin{pmatrix} \th_1(t) F \\ \th_2(t) F \end{pmatrix} := \frac 1 {(2\pi)^d} \int_{\s \in B(\rr)} e^{-it\ls} e^{ix\cdot \s}   \Ps F_\per^\s \, d\s.
\end{equation}

\begin{proposition} \label{prop-reduc-I0}
There exists $C \geq 0$ such that for $t \geq 0$ and $F \in \Sc \times \Sc$ we have
\[
\nr{\Ilowp(t) F - \tIlow(t) F}_\HH \leq C e^{-\g t} \nr F_\LL.
\]
\end{proposition}

\begin{proof}
Let $F \in \Sc \times \Sc$ and $\m \in (0,1]$. We have 
\begin{equation} \label{expr-I0-sigma}
\Ilowp(t) F  = \frac 1 {2i\pi} \frac 1 {(2\pi)^d}  \int_{z \in \G_\m} \int_{\s \in 2\pi \T} \h(\Re(z)) e^{-itz} e^{ix\cdot \s} (\AAs - z)\inv F_\per^\s \, d\s \, dz.
\end{equation}
We write $\Ilowp(t)F = I_{1}(t)F + I_{2}(t)F + I_{3}(t)F$, where $I_{3}(t)F$ is defined as the right-hand side of \eqref{expr-I0-sigma} but with the integral over $\s \in 2\pi\T$ replaced by an integral over $\s \in 2\pi\T \setminus B(\rr)$. For $I_{1}(t)F$ and $I_{2}(t)F$ the integral is taken over $\s \in B(\rr)$. In $I_{1}(t)F$ (in $I_{2}(t)F$, respectively), the function $F_\per^\s$ is replaced by $\Ps F_\per^\s$ (by $(1-\Ps) F_\per^\s$, respectively). Given $\s \in 2\pi\T$, the integrand in \eqref{expr-I0-sigma} is a meromorphic function of $z$ with $\abs {\Re(z)} < 2$ (since $\h(\Re(z))=1$ in this region), and the poles are the eigenvalues of $\AAs$. Thus we can change the contour $\G_\m$ in this region. By Propositions \ref{prop-eigenvalues} and \ref{prop-Floquet-Bloch} we get 
\begin{align*}
\nr{I_{3}(t) F}_{\HH}^2
& \lesssim \nr{\int_{\s \in 2\pi \T \setminus B(\rr)}e^{ix\cdot \s}  \int_{\Im(z) = -\g}  \h(\Re(z)) e^{-itz} (\AAs - z)\inv F_\per^\s\, dz \, d\s }_{\HH}^2 \, d\s\\
& \lesssim \int_{\s \in 2\pi \T \setminus B(\rr)} \nr{ \int_{\Im(z) = -\g}  \h(\Re(z)) e^{-itz} (\AAs - z)\inv F_\per^\s\, dz }_{\LL_\per}^2\, d\s  \\
& \lesssim e^{- 2 \g t} \int_{\s \in 2\pi \T \setminus B(\rr)}  \nr{  F_\per^\s }_{\LL_\per}^2  \, d\s  \\
& \lesssim e^{-2\g t} \nr{F}_{\LL}^2 . 
\end{align*}
We have used the fact that the resolvent $(\AAs-z)\inv$ is uniformly bounded. This is due to the continuity of this resolvent with respect to $z$ and $\s$, by the compactness of the contour of integration, and the compactness of $\bar{2\pi\T \setminus B(\rr)}$.
We similarly have 
\begin{align*}
\nr{I_{2}(t) F}_{\HH}^2
& \lesssim \nr{\int_{\Im(z) = -\g} \int_{\s \in B(\rr)} \h(\Re(z)) e^{-itz} e^{ix\cdot \s} (\AAs - z)\inv (1-\Ps) F_\per^\s \, d\s \, dz}_{\HH}^2\\
& \lesssim e^{-2t\g} \nr{F}_{\LL}^2.
\end{align*}
Now let $\tilde \vf \in C^\infty(\R,(-\ggg_2,-\g])$ be such that $\tilde \vf(\t) = -\g$ if $\abs \t \geq 2$ and $\tilde \vf(\t) \in (-\ggg_2,-\ggg_1)$ if $\abs \t \leq 1$. We set (see Figure \ref{fig-Gamma-tilde}) 
\[
\tilde \G = \set{\t + i \tilde \vf(\t), \t \in \R}.
\]
Then by the residue theorem we have 
\begin{align*}
I_{1}(t) F
& = \frac 1 {2i\pi} \frac 1 {(2\pi)^d} \int_{\s \in B(\rr)}\int_{z \in \G}  \h(\Re(z)) e^{-itz} e^{ix\cdot \s} (\ls - z)\inv  \Ps F_\per^\s \, d\s \, dz\\
& = \tIlow(t) F + \frac 1 {2i\pi}\frac 1 {(2\pi)^d} \int_{\s \in B(\rr)}\int_{z \in \tilde \G}  \h(\Re(z)) e^{-itz} e^{ix\cdot \s} (\ls - z)\inv  \Ps F_\per^\s \, d\s \, dz.
\end{align*}
We estimate the last term as above, and the proof is complete.
\end{proof}

\begin{figure}[ht]
\begin{center}
\resizebox{0.85\textwidth}{!}{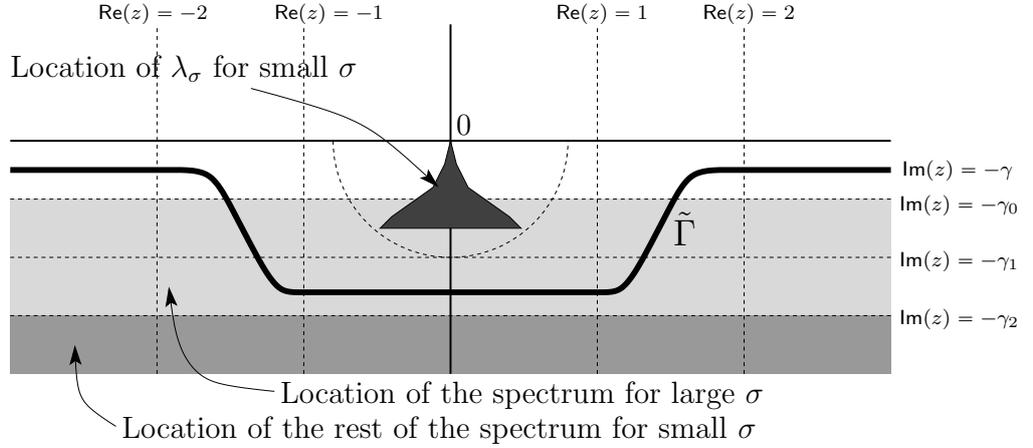}\\[2mm]
\end{center}
\caption{The location of the spectrum of $\AAs$ and the curve $\tilde\Gamma$.}
\label{fig-Gamma-tilde}
\end{figure}

\subsection{Analysis of the first Bloch wave for \texorpdfstring{$\s$}{sigma} small} \label{sec-first-Bloch}

Our purpose is now to estimate $\tIlow(t) F$. For this we describe more precisely the properties of the eigenvalue $\ls$ and the corresponding eigenvector $\Phi_\s$ and eigenprojection $\Ps$ for $\s$ small. We recall that the symmetric matrix $\Gm$ was defined in \eqref{def-G-star}.

\begin{proposition} \label{prop-DL-ls}
The symmetric matrix $\Gm$ is positive and when $\s$ goes to 0 we have
\begin{equation} \label{dec-lambda-s}
\ls = - \frac i {\bm} \innp{\Gm \s}{\s} + \bigo {\abs \s} 0 \big( \abs \s^3 \big).
\end{equation}
Moreover
\begin{equation} \label{dec-phi-s}
\f_\s = \f_0 + i \p_\s +  \bigo {\abs \s} 0 \big( \abs \s^2 \big),
\end{equation}
where $\p_\s \in L^2_\per \cap L^\infty$ is a linear function of $\s$ which satisfies \eqref{def-psi}.
\end{proposition}

\begin{proof}
We first recall that $\ls$ and $\f_\s$ are smooth functions of $\s$, respectively in $\C$ and in $H^k_\per$ for any $k \in \N$. Moreover $\l_0 = 0$ and $\f_0 \equiv 1$. For $\s \in B(\rr)$ we have 
\begin{equation} \label{eq-vp-phi-s}
-(\divg +i \s\trsp)  \Gp(x)  (\nabla + i \s) \f_\s - i \ls \bp \f_\s - \ls^2 \wp \f_\s = 0.
\end{equation}
Taking the inner product with $\f_\s$ gives
\begin{equation} \label{eq-ls}
\innp{\Gp(x) (\nabla + i\s) \f_\s}{(\nabla + i\s) \f_\s} - i \ls \innp{\bp \f_\s}{\f_\s} - \ls^2 \innp{\wp \f_\s}{\f_\s} = 0.
\end{equation}
We take the derivatives of \eqref{eq-ls} with respect to $\s_j$, $j\in \Ii 1 d$, at point $\s = 0$. Since $\innp{\bp \f_0}{\f_0} > 0$ we see that the first derivatives of $\ls$ vanish. Thus, by Taylor expansion, there exists a matrix $Q$ such that 
\[
\ls =  - \frac i {\bm} \innp{Q\s}{\s} + O \big( \abs \s^3 \big).
\]
Since $\s \mapsto \f_\s$ is smooth, we can define $\p_\s \in L^2_\per$ so that \eqref{dec-phi-s} holds. This defines a linear function of $\s$. Taking the linear part in \eqref{eq-vp-phi-s} gives 
\[
- i \divg \Gp(x) (\nabla \p_\s + \s) = 0.
\]
This proves in particular that $\p_\s$ is a solution of \eqref{def-psi}.
Similarly, \eqref{eq-ls} gives 
\[
\innp{\Gp(x) (\nabla \p_\s + \s \f_0)}{(\nabla \p_\s + \s \f_0)} - \frac 1 {\bm} \innp{Q\s}{\s} \innp{\bp \f_0}{\f_0} = O \big( \abs \s^3 \big),
\]
and we deduce
\[
\innp{Q\s}{\s}  = \innp{\Gp(x) (\nabla \p_\s + \s)}{(\nabla \p_\s + \s)} = \innp{\Gm \s}{\s}.
\]
Finally, since $\p_\s$ is periodic its gradient cannot be the constant and non-zero function $-\s$. Therefore $\nabla \p_\s + \s \neq 0$ and hence $\innp{\Gm \s}\s > 0$. This concludes the proof.
\end{proof}

\begin{corollary} \label{cor-estim-ls}
There exist $\L_2 > \L_1 > 0$ such that for $\s \in B(\rr)$
\[
\L_1 \abs \s^2 \leq \Re(-i\ls) \leq \L_2 \abs \s^2,
\]
and
\[
\L_1 \abs \s^2 \leq \frac {\innp{\Gm \s}{\s}}{\bm} \leq \L_2 \abs \s^2.
\]
\end{corollary}

Now we describe more precisely the projection $\Ps$.

\begin{proposition} \label{prop-Psi}
There exists $\Psi_\s \in \LL_\per$ which depends smoothly on $\s \in B(\rr)$ and such that for $\s \in B(\rr)$ and $F \in \LL_\per$ we have 
\[
\Ps F = \innp{F}{\Psi_\s}_{\LL_\per} \Phi_\s.
\]
Moreover
\begin{equation} \label{def-Psi0}
\Psi_0 = \frac 1 {\bm} \begin{pmatrix} \bp \\ i \end{pmatrix}.
\end{equation}
\end{proposition}

\begin{proof}
Let $\s \in B(\rr)$. Since $\Ps$ is the projection on the line spanned by $\Phi_\s$ we have, for all $F \in \LL_\per$,
\[
\Ps F = \innp{\Ps F}{\Phi_\s} \Phi_\s.
\]
Since $F \mapsto \innp{\Ps F}{\Phi_\s}_{\LL_\per}$ is a continuous linear form on $\LL_\per$ which depends smoothly on $\s$, the first statement follows from the Riesz representation theorem. 

The adjoint of $\AAso$ in $\LL_\per$ is 
\[
\AAsso  =
\begin{pmatrix}
0 & \Pp \\
\wp\inv & i\ap
\end{pmatrix}.
\]
For $F \in H^2_\per \times H^1_\per$ we have 
\[
\innp{\AAso F}{\Psi_0}_{\LL_\per} \Phi_0 = \Pso \AAso F = \AAso \Pso F = 0.
\]
This proves that $\Psi_0 \in \Dom(\AAsso)$ and $\AAsso \Psi_0 = 0$. We can check by direct computation that this implies that there exists $\a \in \C$ such that $\Psi_0 = \a \begin{pmatrix} \bp \\ i \end{pmatrix}$. Since
\[
1 = \innp{\Phi_0}{\Psi_0} = \a \int_{\T} \bp,
\]
we have $\a = \bm\inv$, and the proof is complete.
\end{proof}

\begin{remark} \label{rem-reg-Psi}
Since $F \mapsto \innp{\Ps F}{\Phi_\s}_{\LL_\per}$ is also a smooth function in $\mathscr{L}(L^2_\per \times H\inv_\per , \LL_\per)$ we can also see $\Psi_\s$ as a smooth function of $\s$ in $L^2_\per \times H^1_\per$.
\end{remark}

\subsection{Comparison between the periodic wave equation and the heat equation}

In this paragraph we prove Theorem \ref{th-heat}. Given $F = (u_0,i w u_1) \in \Sc \times \Sc$, we denote by $\um(t)$ the solution of the heat problem \eqref{heat-Q}-\eqref{heat-Q-CI}. Our purpose is to compare the solution $\upp(t)$ of \eqref{wave-per} with $\um(t)$. We set 
\[
v_0 = \frac {\bp  u_0 + \wp u_1} {\bm},
\]
and we denote by $\hat v_0$ the Fourier transform of $v_0$. We first recall that the decay of $\um(t)$ is also governed by the contribution of low frequencies. 

\begin{lemma} \label{lem-uQ}
Let $\rr > 0$ be given by Proposition \ref{prop-eigenvalues}. Then there exists $\tilde \g > 0$ such that for $t \geq 1$ we have in $L^2(\R^d)$
\begin{eqnarray*}
\um(t) &=& \frac 1 {(2\pi)^d} \int_{\x \in B(\rr)} e^{i x\cdot \x} e^{-\frac {t}{\bm} \innp{\Gm \x}{\x}} \widehat{v_0}(\x)\, d\x + \Oc\big( e^{-\tilde \g t} \big) \nr{v_0}_{L^2(\R^d)},\\
\nabla \um(t) &=& \frac 1 {(2\pi)^d} \int_{\x \in B(\rr)} i \x e^{i x\cdot \x} e^{-\frac {t}{\bm} \innp{\Gm \x}{\x}} \widehat{v_0}(\x)\, d\x + \Oc\big( e^{-\tilde \g t} \big) \nr{v_0}_{L^2(\R^d)},\\
i \wp \partial_t \um(t) &=& -\frac {i\wp} {(2\pi)^d} \int_{\x \in B(\rr)} e^{i x\cdot \x} \frac {\innp{\Gm \x}{\x}}{\bm} e^{-\frac {t}{\bm} \innp{\Gm \x}{\x}} \widehat{v_0}(\x)\, d\x + \Oc\big( e^{- \tilde \g t} \big) \nr{v_0}_{L^2(\R^d)}.
\end{eqnarray*}
\end{lemma}

\begin{proof}
We prove for instance the second estimate. The others are similar. For $t\geq 1$ and $x \in \R^d$ we have 
\begin{align*}
\nabla \um(t) = \nabla e^{-\frac {t \Pm} {\bm}} v_0 = \frac 1 {(2\pi)^d} \int_{\x \in \R^d} i \x e^{i x\cdot \x} e^{-\frac {t}{\bm} \innp{\Gm \x}{\x}} \widehat{v_0}(\x)\, d\x.
\end{align*}
By Corollary \ref{cor-estim-ls} we have 
\[
\abs{ \int_{\x \in \R^d \setminus B(\rr)} i \x e^{i x\cdot \x} e^{-\frac {t}{\bm} \innp{\Gm \x}{\x}} \widehat{v_0}(\x)\, d\x} \leq \int_{\x \in \R^d \setminus B(\rr)} \abs \x e^{- t \L_1 \abs \x^2 } \abs{\widehat{v_0}(\x)} \, d\x.
\]
The estimate then follows from the Cauchy-Schwarz inequality and the Plancherel equality.
\end{proof}

Theorem \ref{th-heat} is a consequence of Propositions \ref{prop-eAAc-Ilow} and \ref{prop-reduc-I0} together with the following estimates. We recall that $\th_1(t)$ and $\th_2(t)$ were defined in \eqref{def-tilde-Ilow}. Moreover, we recall that by density it is enough to prove Theorem \ref{th-heat} for $U_0 = F \in \Sc \times \Sc$.

\begin{proposition} \label{prop-decay-uv}
Let $s_1,s_2 \in \big[0,\frac d 2 \big]$ and $\k > 1$. Then there exists $C \geq 0$ which does not depend on $F \in \Sc \times \Sc$ and such that for $t \geq 1$ we have 
\[
\nr{\th_1(t) F - \um(t)}_{L^{2,-\k s_1}} \leq  C \pppg t^{- \frac 1 2 - \frac {s_1 + s_2} 2} \nr{F}_{\LL^{\k s_2}},
\]
\[
\nr{\nabla \th_1(t) F - W \nabla \um(t)}_{L^{2,-\k s_1}} \leq  C \pppg t^{- 1 - \frac {s_1 + s_2} 2} \nr{F}_{\LL^{\k s_2}},
\]
and
\[
\nr{\th_2(t) F - i \wp \partial_t \um(t)}_{L^{2,-\k s_1}} \leq  C \pppg t^{- \frac 3 2 - \frac {s_1 + s_2} 2} \nr{F}_{\LL^{\k s_2}}.
\]
\end{proposition}

\begin{proof}
For $j \in \{1,2\}$ we set $p_j = \frac {2d}{d-2s_j} \in [2,+\infty]$. Then $p_0 \in [1,+\infty]$ is defined by 
\[
\frac 1 {p_0} + \frac 1 {p_1} + \frac 1 {p_2} = 1.
\]
We begin with the last estimate. By Propositions \ref{prop-DL-ls} and \ref{prop-Psi}, and \eqref{Phi-phi}, we have 
\[
\th_2(t) F = \frac 1 {(2\pi)^d} \int_{\s \in B(\rr)}  e^{ix\cdot \s} e^{-\frac t {\bm} \left({\innp{\Gm \s}{\s}} + \Oc(\abs \s^3)\right)} \innp{F_\per^\s}{\Psi_\s}_{\LL_\per}  \ls \wp \f_\s \, d\s.
\]
Let 
\[
v_1(t) = \frac 1 {(2\pi)^d} \int_{\s \in B(\rr)} e^{ix\cdot \s} e^{-\frac t {\bm} {\innp{\Gm \s}{\s}}}  \innp{F_\per^\s}{\Psi_\s} \ls \wp \f_\s \, d\s.
\]
For $g \in \Sc$ we have by Proposition \ref{prop-Floquet-Bloch}
\[
 \innp{e^{ix\cdot \s} \wp \f_\s}{g}_{L^2(\R^d)} = \innp{\wp \f_\s}{g_\per^\s}_{L^2_\per},
\]
so
\begin{eqnarray*}
\lefteqn{\abs{\innp{\th_2(t) F - v_1(t)}{g}_{L^2(\R^d)}}}\\
&& = \abs{\frac 1 {(2\pi)^d} \int_{\s \in B(\rr)}  \ls  e^{-\frac t {\bm} {\innp{\Gm \s}{\s}}} \left(e^{t \Oc(\abs \s^3)} -1\right) \innp{F_\per^\s}{\Psi_\s}_{\LL_\per}  \innp{\wp \f_\s}{g_\per^\s}_{L^2_\per} \, d\s}\\
&& \lesssim \int_{\s \in B(\rr)} |\s|^2 e^{- \L_1 t \abs \s^2} \big|e^{t\Oc(\abs \s^3)} -1\big| \, \big|\innp{F_\per^\s}{\Psi_\s}\big| \, \big|\innp{\wp \f_\s}{g_\per^\s}\big| \, d\s\\
&& \lesssim \int_{\s \in B(\rr)}   t \abs \s^5 e^{- \L_1 t \abs \s^2} e^{t\Oc(\abs \s^3)} \big|\innp{F_\per^\s}{\Psi_\s}\big| \, \big|\innp{\wp \f_\s}{g_\per^\s}\big| \, d\s.
\end{eqnarray*}
Choosing $r > 0$ smaller if necessary we obtain 
\[
\abs{\innp{\th_2(t) F - v_1(t)}{g}_{L^2(\R^d)}} \lesssim \int_{\s \in B(\rr)}   t \abs \s^5 e^{- \frac {\L_1 t \abs \s^2}2} \big|\innp{F_\per^\s}{\Psi_\s} \big|  \, \big|\innp{\wp \f_\s}{g_\per^\s}\big| \, d\s.
\]
By the H\"older inequality  we have
\begin{multline*}
\abs{\innp{\th_2(t) F - v_1(t)}{g}_{L^2(\R^d)}}\\
\lesssim 
t^{-\frac 32} \Big\|(t \abs \s^2)^{\frac 5 2} e^{-\frac {t\L_1 \abs \s^2}2}\Big\|_{L^{p_0}_\s(B(\rr))} 
\nr{\innp{\wp \f_\s}{g_\per^\s}}_{L^{p_1}_\s(B(\rr))}\nr{\innp{F_\per^\s}{\Psi_\s}}_{L^{p_2}_\s(B(\rr))}.
\end{multline*}
If $p_0 = \infty$ (\ie if $s_1 + s_2 = 0$) then 
\[
\Big\| (t \abs \s^2)^{\frac 5 2} e^{-\frac {t\L_1 \abs \s^2}2} \Big\| _{L^{p_0}_\s(B(\rr))} \lesssim 1.
\]
And if $p_0 < \infty$,
\begin{align*}
\Big\| (t \abs \s^2)^{\frac 5 2} e^{-\frac {t\L_1 \abs \s^2}2}\Big\|_{L^{p_0}_\s(B(\rr))} = \left(t^{-\frac {d}2} \int_{B(\sqrt t r)} \abs{\y}^{5p_0} e^{-\frac {p_0 \L_1 \abs{\y}^2}2} \, d\y\right)^{\frac 1 {p_0}} \lesssim t^{-\frac d {2p_0}} = t^{-\frac {s_1 + s_2}2}.
\end{align*}
By Corollary \ref{cor-Floquet-Bloch} we finally get in both cases
\[
\abs{\innp{\th_2(t) F - v_1(t)}{g}_{L^2(\R^d)}} \lesssim \pppg t^{- \frac 32 -\frac {s_1 + s_2}2} \nr{g}_{L^{2,\k s_1}} \nr{F}_{\LL^{\k s_2}}.
\]

In $\LL_\per$ we have $\Psi_\s = \Psi_0 + \Oc (\abs \s)$ so, if we set 
\[
v_2(t) = \frac 1 {(2\pi)^d} \int_{\s \in B(\rr)} e^{ix\cdot \s} e^{-\frac t {\bm} {\innp{\Gm \s}{\s}}}  \innp{F_\per^\s}{\Psi_0}_{\LL_\per} \ls \wp \f_\s \, d\s,
\]
then we similarly obtain
\begin{align*}
\abs{\innp{v_1(t) - v_2(t)}{g}_{L^2(\R^d)}}
& \lesssim \int_{\s \in B(r)} \abs \s^2 e^{-t \L_1 \abs \s^2} \abs{\innp{\wp \f_\s}{g_\per^\s}} \abs{\innp{F_\per^\s}{ {\Psi_\s-\Psi_0}}} \, d\s\\
& \lesssim \int_{\s \in B(\rr)} \abs \s^3 e^{-t \L_1 \abs \s^2} \abs{\innp{\wp \f_\s}{g_\per^\s}} \abs{\innp{F_\per^\s}{\frac {\Psi_\s- \Psi_0}{\abs \s}}} \, d\s \\
& \lesssim \pppg t^{- \frac 3 2 - \frac {s_1 + s_2} 2} \nr{g}_{L^{2,\k s_1}} \nr{F}_{\LL^{\k s_2}}.
\end{align*}
Similarly,
\[
\ls \wp \f_\s = - \frac i {\bm} \innp{\Gm\s}{\s} \wp   + \Oc\big( \abs \s^{3} \big),
\]
so 
\[
\abs{\innp{v_2(t) - v_3(t)}{g}_{L^2(\R^d)}} \lesssim \pppg t^{- \frac 3 2 - \frac {s_1 + s_2} 2} \nr{g}_{L^{2,\k s_1}} \nr{F}_{\LL^{\k s_2}}
\]
where we have set 
\[
v_3(t) = - \frac i {\bm} \frac 1 {(2\pi)^d} \int_{\s \in B(\rr)} e^{ix\cdot \s} e^{-\frac t {\bm} {\innp{\Gm \s}{\s}}} \innp{\Gm\s}{\s} \wp  \innp{F_\per^\s}{\Psi_0}  \, d\s.
\]
Finally, by \eqref{def-Psi0} and Proposition \ref{prop-Floquet-Bloch} we have 
\begin{align*}
\innp{F_\per^\s}{\Psi_0}
& = \frac 1 {\bm} \left(\innp{(u_0)_\per^\s}{\bp} + \innp{i(\wp u_1)_\per^\s}{i} \right)\\
& = \frac 1 {\bm} \int_{x \in \R^d} e^{-ix\cdot \s} \big(\bp(x) u_0(x) + \wp u_1(x) \big) \, dx\\
& = \widehat{v_0} (\s).
\end{align*}
With Lemma \ref{lem-uQ} we have $\nr{v_3(t) - i\wp \partial_t \um(t)} = \Oc\big( e^{-\tilde \g t} \nr {F} \big)$, which concludes the proof of the third estimate. For the first estimate, we proceed similarly except that $\ls \wp \f_\s$ is replaced by $\f_\s$. 
For the second we start from
\[
\nabla \th_2(t) F = \frac 1 {(2\pi)^d} \int_{\s \in B(\rr)} e^{ix\cdot \s} e^{-\frac t {\bm} \left({\innp{\Gm \s}{\s}} + \Oc(\abs \s^3)\right)} \innp{F_\per^\s}{\Psi_\s}_{\LL_\per}  \big(i\s\f_\s + \nabla \f_\s  \big) \, d\s.
\]
By Proposition \ref{prop-DL-ls} and \eqref{def-J} we have 
\begin{equation} \label{eq-der-phi}
i\s\f_\s + \nabla \f_\s = i\s + i \nabla \p_\s + O \big( \abs \s^2 \big) = i W(x) \s + O \big( \abs \s^2 \big),
\end{equation}
so we can proceed as above to get the second estimate and conclude the proof.
\end{proof}

\section{Low frequency analysis in the perturbed setting} \label{sec-perturbed}

In this section we prove Theorem \ref{th-diff-decay}. By Proposition \ref{prop-eAAc-Ilow}, it is enough to estimate the difference between $\Ilow(t)$ and $\Ilowp(t)$ (defined by \eqref{def-Ilow} and \eqref{def-Ilowp}, respectively). Since the perturbation breaks the periodic structure, it is no longer possible to reduce the analysis to a family of problems on the torus. Here, we will deduce the time decay from resolvent estimates. We recall that the contour $\G_\m$ was defined in \eqref{def-Gamma}.\\

We start from 
\[
(\Ilow(t) - \Ilowp(t)) F = \int_{\R} \h(\t) e^{-it(\t + i \th_\m(\t))} \Rc(\t + i \th_\m(\t))  F \, d\t,
\]
where 
\[
\Rc(z) := \frac 1 {2i\pi} \big((\AAc-z)\inv - (\AAp-z)\inv\big).
\]
By partial integrations we obtain, for all $k \in \N$,
\[
(\Ilow(t) - \Ilowp(t)) F = \frac 1 {(it)^k}  \int_{\R} e^{-it(\t+i\th_\m(\t))} L_\m^k \big(\h(\t) \Rc(\t + i \th_\m(\t))  \big) F \, d\t,
\]
where 
\[
L_\m := \frac {d} {d\t} \frac 1 {1 + i\th_\m'(\t)}.
\]
We recall that $(\Ilow(t) - \Ilowp(t)) F$ does not depend on $\m \in (0,1]$. However, if we assume that the derivatives of $\th_\m$ are bounded uniformly in $\m$, the estimates given by this equality are of the form 
\[
\nr{(\Ilow(t) - \Ilowp(t)) F} \lesssim \frac {e^{\m t}}{t^k} \sup_{\abs \t \leq 3} \sup_{0\leq j \leq k} \nr{\Rc^{(j)} (\t + i \th_\m(\t))}.
\]
In $\mathscr{L}(L^2(\R^d))$, the resolvents blow up near 0, so we cannot simply let $\m$ go to 0 to get rid of the exponential factor. However, it is standard in this kind of contexts that in suitable weighted spaces some derivatives of these resolvents can be uniformly bounded. In this section, we prove uniform estimates for the derivatives of $\Rc$ in weighted spaces. Then, at the limit $\m \to 0$, this will give polynomial decay for the difference $(\Ilow(t) - \Ilowp(t)) F$, hence for the difference $u(t)-\upp(t)$ as in Theorem \ref{th-diff-decay}.\\

For $\b_x \in \N^d$ we set 
\[
\Thbxo := 
\begin{pmatrix}
\partial_x^{\b_x} & 0 \\
0 & 0
\end{pmatrix}, 
\quad  
\Th_1^0 := \begin{pmatrix} 0 & 0 \\ 0 & 1 \end{pmatrix} 
\quad \text{and} \quad  
\Thoo := \begin{pmatrix} 1 & 0 \\ 0 & 0 \end{pmatrix}.
\]
We recall that the solution $e^{-it\AAc} U_0$ of \eqref{wave-A} is of the form $\big(u(t),i w \partial_t u(t) \big)$ where $u(t)$ is the solution of \eqref{wave}. Thus, for $\d \in \R$, $\b_t \in \{0,1\}$ and $\b_x \in \N^d$ such that $\b_t + \abs {\b_x} \leq 1$ we have 
\[
\nr{\Thbxbt e^{-it\AAc} U_0}_{\LL^{-\d}} \simeq \nr{\partial_t^{\b_t} \partial_x^{\b_x} u(t)}_{L^{2,-\d}(\R^d)}.
\]
For $\tilde \b_x \in \N^d$ we also set
\[
\tThbx := 
\begin{pmatrix}
0 & 0\\
\partial_x^{\tilde \b_x}& 0
\end{pmatrix}
\quad \text{and} \quad 
\widetilde \Th_{0} := \Id_{\LL} = \begin{pmatrix} 1 & 0 \\ 0 & 1 \end{pmatrix}.
\]
This odd notation will prove to be useful in the sequel.

\subsection{Resolvent estimates in the periodic case}

In order to prove estimates on the derivatives of $(\AAc-z)\inv$ and of the difference $(\AAc-z)\inv - (\AAp-z)\inv$, we need more information about the resolvent of $\AAp$.\\

\begin{proposition} \label{prop-res-Rp}
Let $\b_t \in \{0,1\}$ and $\b_x, \tilde \b_x \in \N^d$ with $\b_t + \abs {\b_x} \leq 1$ and $|\tilde \b_x| \leq 1$. Let $s_1,s_2 \in \big[0,\frac d 2 \big)$ and $\k > 1$. Then there exist a neighborhood $\Uc$ of 0 in $\C$ and $C \geq 0$ such that for $z \in \C_+ \cap \Uc$ and $m \in \N$ we have 
\[
\nr{\Thbxbt (\AAp-z)^{-1-m} \tThbx}_{\mathscr{L}(\LL^{\k s_2},\LL^{-\k s_1})} \leq C \left( 1 + \abs z^{\b_t + \frac {\abs {\b_x} + \abs{\tilde \b_x}} 2 + \frac {s_1 + s_2}{2} - (1 + m)}\right).
\]
\end{proposition}

\begin{proof}
We set $\tilde \k = \sqrt \k$. Without loss of generality we can assume that $\k$ is so close to 1 that $\tilde \k \max(s_1,s_2) \leq \frac d 2$. We follow the same ideas as for the propagator. For this we can still use the Floquet-Bloch decomposition. Thus, for $\s \in \R^d$ we set
\[
\Thbxbt(\s) = e^{-ix\cdot\s} \Thbxbt e^{ix\cdot\s}.
\]
We similarly define $\tThbx(\s)$. Let $F,G \in \Sc \times \Sc$. We can write
\begin{eqnarray*}
\lefteqn{\innp{\Thbxbt (\AAp-z)^{-1-m} \tThbx F}{G}_\LL}\\
&& = \frac 1 {(2\pi)^d} \int_{\s \in 2\pi \T} \innp{\Thbxbts (\AAs -z)^{-1-m} \tThbxs F_\per^\s}{G_\per^\s}_{\LL_\per} \, d\s\\
&& = \frac 1 {(2\pi)^d} \int_{\s \in B(r)} \frac {1}{(\ls - z)^{1+m}}  \innp{\Thbxbts \Ps \tThbxs F_\per^\s}{G_\per^\s}_{\LL_\per} \, d\s + \innp{\Bc(z)F}{G}_\LL,
\end{eqnarray*}
where $\Bc \in \mathscr{L}(\LL)$ is holomorphic in a neighborhood of 0 in $\C$ and 
\[
\innp{\Thbxbts \Ps \tThbxs F_\per^\s}{G_\per^\s}_{\LL_\per} = \innp{F_\per^\s}{\tThbxs^* \Psi_\s}_{\LL_\per} \innp{\Thbxbts \Phi_\s}{G_\per^\s}_{\LL_\per}.
\]
By Proposition \ref{prop-DL-ls} and the fact that $\Im(z) > 0$ we have
\[
\frac 1 {\abs {\ls-z}} \lesssim \frac 1 {\L_1 \abs \s^2 + \abs z}.
\]
By Remark \ref{rem-reg-Psi} and the expression of $\Psi_0$ in Proposition \ref{prop-Psi} we have
\[
\nr{\tThbxs^* \Psi_\s}_{\LL_\per} = \Oc \Big(\abs \s^{\abs{\tilde \b_x}}\Big).
\]
And we recall from \eqref{Phi-phi}, Corollary \ref{cor-estim-ls} and \eqref{eq-der-phi} that 
\[
\nr{\Thbxbts \Phi_\s}_{\LL_\per} = \Oc \Big(\abs \s^{2 \b_t + \abs{\b_x}}\Big).
\]
For $j \in \{1,2\}$ we set $p_j = ({2d})/({d-2 \tilde \k s_j})$. Then we consider $p_0 \in [1,+\infty]$ such that 
\[
\frac 1 {p_0} + \frac 1 {p_1} + \frac 1 {p_2} = 1.
\]
By the H\"older inequality and Corollary \ref{cor-Floquet-Bloch} (applied with $\tilde \k$ instead of $\k$) we get 
\begin{eqnarray*}
\lefteqn{\abs{\innp{\Thbxbt (\AAp-z)^{-1-m} \tThbx F}{G}_\LL}}\\
&& \lesssim \nr{\frac {\abs \s^{2\b_t + |\b_x| + |\tilde \b_x|}}{(\L_1 \abs \s^2 + \abs z)^{1+m}}}_{L^{p_0}_\s}
\nr{\frac {\innp{F_\per^\s}{\tThbxs^* \Psi_\s}}{\abs{\s}^{\abs{\tilde \b_x}}}}_{L^{p_2}_\s}
\nr{\frac {\innp{\Thbxbts \Phi_\s}{G_\per^\s}}{\abs{\s}^{2 \b_t + \abs{\b_x}}}}_{L^{p_1}_\s}\\
&& \lesssim \nr{\frac {\abs \s^{2\b_t + |\b_x| + |\tilde \b_x|}}{(\L_1 \abs \s^2 + \abs z)^{1+m}}}_{L^{p_0}_\s(B(\rr))} \nr{F}_{L^{2,\k s_2}(\R^d)} \nr{G}_{L^{2,\k s_1}(\R^d)}.
\end{eqnarray*}
We have 
\[
\sup_{\s \in B(\rr)} \frac {\abs \s^{2\b_t + |\b_x| + |\tilde \b_x|}}{(\L_1 \abs \s^2 + \abs z)^{1+m}} \lesssim  1 + \abs z^{\frac {2\b_t + |\b_x| + |\tilde \b_x|}2- (1+m)},
\]
so the proposition is proved if $s_1 = s_2 = 0$ and $p_0 = \infty$.
Now assume that $s_1 + s_2 > 0$. Using polar coordinates in $\s$ we can write 
\[
\nr{\frac {\abs \s^{2\b_t + |\b_x| + |\tilde \b_x|}}{(\L_1 \abs \s^2 + \abs z)^{1+m}}}_{L^{p_0}_\s(B(\rr))}^{p_0}
\lesssim \int_0^\rr \frac {\th^{(2\b_t + |\b_x| + |\tilde \b_x|)p_0}}{(\L_1 \th^2 + \abs z)^{(1+m)p_0}} \th^{d-1} \, d\th.
\]
If $\big( 2\b_t + |\b_x| + |\tilde \b_x| - 2(1+m) \big) p_0 + d-1 > - 1$ then this quantity is bounded uniformly in $z \in \C_+$ close to 0. Otherwise, the change of variables $\th = \sqrt {\abs z} \tilde \th$ gives
\begin{multline*}
\nr{\frac {\abs \s^{2\b_t + |\b_x| + |\tilde \b_x|}}{(\L_1 \abs \s^2 + \abs z)^{1+m}}}_{L^{p_0}_\s(B(\rr))}^{p_0}\\
\leq \abs z^{\frac {(2 \b_t + |\b_t| + |\tilde \b_x|)p_0 + d}{2} - (1+m)p_0} \int_0^{\frac r {\sqrt{\abs z}}}  \frac {\tilde \th^{(2 \b_t + |\b_t| + |\tilde \b_x|  )p_0+ d-1}} {(\L_1 \tilde \th^2 + 1)^{(1+m)p_0}} \, d\tilde \th.
\end{multline*}
In any case we can write the rough estimate  
\[
 \nr{\frac {\abs \s^{2\b_t + |\b_x| + |\tilde \b_x|}}{(\L_1 \abs \s^2 + \abs z)^{1+m}}}_{L^{p_0}_\s(B(\rr))} \lesssim  1 + \abs z^{- (1+m) + \frac {2\b_t + |\b_x| + |\tilde \b_x|} 2 + \frac d {2 p_0} - \e},
\]
where
\[
\e = \frac {(\tilde \k -1)(s_1+s_2)}2 > 0
\] 
(in fact we can take $\e = 0$ if $\big( 2\b_t + |\b_x| + |\tilde \b_x| - 2(1+m) \big) p_0 + d-1 < - 1$).
Since $d/p_0 = \tilde \k(s_1 + s_2)$, the conclusion follows.
\end{proof}

\subsection{Resolvent estimates in the perturbed setting}

In this paragraph we use the estimate of the derivatives of $(\AAp-z)\inv$ for $z \in \C_+$ close to 0 to obtain (better) estimates for the difference $(\AAc-z)\inv - (\AAp-z)\inv$. This will prove that we have the same estimates for $(\AAc-z)\inv$ as for $(\AAp-z)\inv$.

\begin{proposition} \label{prop-res-general}
Let $\b_t \in \{0,1\}$ and $\b_x \in \N^d$ with $\b_t + \abs {\b_x} \leq 1$. Let $s_1,s_2 \in \big[0,\frac d 2 \big)$ and $\k > 1$. Let $\y > 0$. Assume that 
\begin{equation} \label{hyp-eta}
\k \max(s_1,s_2) + \k \y < \min\left(\frac d 2, \rho_G, \rho_a + 1 \right).
\end{equation}
Then there exists $C \geq 0$ such that for $z \in \C_+$ with $\abs z \leq 1$ we have 
\begin{multline*}
\nr{\Thbxbt \big((\AAc-z)^{-1-m} - (\AAp-z)^{-1-m} \big)}_{\mathscr{L}(\LL^{\k s_2},\LL^{-\k s_1})}\\
\leq C \left( 1 + \abs z^{\frac \y 2 + \b_t + \frac {\abs {\b_x}} 2 + \frac {s_1 + s_2}2 - (1 + m)}\right).
\end{multline*}
\end{proposition}

We split the proof of this proposition into several intermediate results. We begin with a remark which will be used several times in the proofs. It is based on the fact that in the expression of the resolvent in Proposition \ref{prop-Ac-max-diss} the lower row is, up to a term $w$, equal to $zw$ times the upper row.

\begin{remark} \label{rem-colonne-droite}
Let $\nu_1,\nu_2 \in L^\infty$. Then for $z \in \C_+$ we have 
\begin{equation} \label{eq-colonne-droite}
\begin{pmatrix}
0 & \nu_1 \\
0 & \nu_2 
\end{pmatrix}
(\AAc-z)\inv =
\begin{pmatrix}
\n_1 w & 0 \\
\n_2 w & 0
\end{pmatrix}
+ z 
\begin{pmatrix}
\n_1w & 0 \\
\n_2w & 0 
\end{pmatrix}
(\AAc-z)\inv.
\end{equation}
This also holds with $\AAc$ replaced by $\AAp$. In particular for $\d_1 \in \R$ and $U \in \Sc \times \Sc$ we have 
\begin{multline*}
\nr{\begin{pmatrix}
0 & \nu_1 \\
0 & \nu_2 
\end{pmatrix}
\big( (\AAc-z)\inv - (\AAp-z)\inv \big) U}_{\LL^{-\d_1}}\\
\lesssim \nr{U}_{\LL^{-\d_1}} + \abs z \nr{\big( (\AAc-z)\inv - (\AAp-z)\inv \big) U}_{\LL^{-\d_1}}.
\end{multline*}
Now if we take the derivatives of \eqref{eq-colonne-droite} with respect to $z$ we get for $m \geq 1$
\begin{eqnarray*}
\lefteqn{
\nr{\begin{pmatrix}
0 & \nu_1 \\
0 & \nu_2 
\end{pmatrix}
\big( (\AAc-z)^{-1-m} - (\AAp-z)^{-1-m} \big) U}_{\LL^{-\d_1}}}\\
&& \lesssim \nr{\big( (\AAc-z)^{-m} - (\AAp-z)^{-m} \big)U}_{\LL^{-\d_1}}\\
&& + \abs z \nr{\big( (\AAc-z)^{-1-m} - (\AAp-z)^{-1-m} \big) \big) U}_{\LL^{-\d_1}}.
\end{eqnarray*}
\end{remark}

We can apply this remark in particular to the operator $\Thuo$ which select the component $i w \partial_t u(t)$ in the solution of \eqref{wave-A}:

\begin{lemma} \label{lem-prop-res-bt}
Assume that the result of Proposition \ref{prop-res-general} holds when $\b_t = 0$. Then it also holds when $\b_t = 1$.
\end{lemma}

\begin{proof}
Assume that $m \geq 1$. By Remark \eqref{rem-colonne-droite} we have in $\mathscr{L}(\LL^{\k s_2},\LL^{-\k s_1})$
\begin{equation*}
\begin{aligned}
\nr{\Thuo\big((\AAc-z)^{-1-m} - (\AAp-z)^{-1-m} \big)}
& \lesssim \nr{\Thoo \big((\AAc-z)^{-m} - (\AAp-z)^{-m} \big)}\\
& + \abs z \nr{\Thoo \big((\AAc-z)^{-1-m} - (\AAp-z)^{-1-m} \big)},
\end{aligned}
\end{equation*}
and the conclusion for $\b_t = 1$ follows from the case $\b_t = 0$. We conclude similarly if $m = 0$.
\end{proof}

After Lemma \ref{lem-prop-res-bt} it is enough to consider the case $\b_t = 0$. For this we will use perturbation arguments. We set 
\[
P_0 = -\divg G_0(x) \nabla.
\]
Then we write 
\[
\AAc - \AAp = \Pc_0 + \Dc_0,
\]
where 
\[
\Pc_0 = 
\begin{pmatrix}
0 & 0 \\
\Po & 0
\end{pmatrix}
\quad \text{and} \quad 
\Dc_0 =
\begin{pmatrix}
0 & w\inv - \wp\inv \\
0 & \ao
\end{pmatrix}.
\]
Notice that $w\inv - \wp\inv$ has the same decay property as $w_0$ in \eqref{hyp-vanishing}.\\

We begin with the contribution of $\Pc_0$. For this we set 
\[
\tAAc = \AAp + \Pc_0.
\]
Notice that all the general results proved for $\AAc$ in Section \ref{sec-spectral} also apply for $\tAAc$.

\begin{lemma} \label{lem-estim-P0}
In the setting of Proposition \ref{prop-res-general}, if $\tilde \b_x \in \N^d$ is such that $|\tilde \b_x| \leq 1$ then there exists $C \geq 0$ such that for $z \in \C_+$ with $\abs z \leq 1$ we have 
\[
\nr{\Thbxo (\tAAc-z)^{-1} \tThbx}_{\mathscr{L}(\LL^{\k s_2},\LL^{-\k s_1})} \leq C \left( 1 + \abs z^{\frac { \abs {\b_x} + |\tilde \b_x| + s_1 + s_2}2 - 1}\right).
\]
\end{lemma}

\begin{proof}
For $z \in \C_+$ we set 
\[
\Rp(z) = \big(\Pp -iz\bp(x)-z^2\wp(x)\big)\inv
\]
and
\[
R_0(z) = \big(\PG -iz\bp(x)-z^2\wp(x)\big)\inv.
\]
We observe from Proposition \ref{prop-Ac-max-diss} that 
\begin{equation} \label{equiv-Ac-Ro}
\nr{\Thbxo (\tAAc-z)^{-1} \tThbx}_{\mathscr{L}(\LL^{\k s_2},\LL^{-\k s_1})} \simeq \nr{\pppg x^{-\k s_1} \partial_x^{\b_x} R_0(z) \partial_x^{\tilde \b_x} \pppg x^{-\k s_2}}_{\mathscr L (L^2(\R^d))}.
\end{equation}
We have a similar estimate with $\tAAc$ and $R_0(z)$ replaced by $\AAp$ and $\Rp(z)$, respectively.

Let $\s_1,\s_2 \in \big[ 0 ,\frac d 2 \big[$. Let $\b_1,\b_2 \in \N^d$ with $\abs {\b_1} \leq 1$ and $\abs{\b_2} \leq 1$.
For $\vf \in \Sc$ we have 
\begin{eqnarray*}
\lefteqn{\nr{\nabla R_0(z) \partial_x^{\b_2} \pppg x^{-\k \s_2} \vf}_{L^2(\R^d)}^2}\\
&&\lesssim \innp{G(x) \nabla R_0(z)  \partial_x^{\b_2} \pppg x^{-\k \s_2} \vf}{\nabla R_0(z)  \partial_x^{\b_2} \pppg x^{-\k \s_2} \vf}\\
&& \lesssim \innp{\PG R_0(z) \partial_x^{\b_2} \pppg x^{-\k \s_2}  \vf}{R_0(z)  \partial_x^{\b_2} \pppg x^{-\k \s_2} \vf} \\
&& \lesssim \abs{\innp{\vf}{\pppg x^{-\k \s_2}\partial_x^{\b_2} R_0(z)\partial_x^{\b_2} \pppg x^{-\k \s_2} \vf}}\\
&& \quad + \abs{\innp{(iz\bp + z^2 \wp) R_0(z) \partial_x^{\b_2} \pppg x^{-\k \s_2}\vf}{R_0(z) \partial_x^{\b_2} \pppg x^{-\k \s_2}\vf}},
\end{eqnarray*}
hence
\begin{multline}\label{estim-nabla-Ro}
\nr{\nabla R_0(z) \partial_x^{\b_2} \pppg x^{-\k \s_2}}_{\mathscr{L}(L^2(\R^d))}\\
\lesssim \nr{\pppg x^{-\k \s_2}\partial_x^{\b_2} R_0(z)\partial_x^{\b_2} \pppg x^{-\k \s_2}}_{\mathscr{L}(L^2(\R^d))}^{\frac 12} + \abs z^{\frac 12} \nr{R_0(z)  \partial_x^{\b_2} \pppg x^{-\k \s_2}}_{\mathscr{L}(L^2(\R^d))}.
\end{multline}
On the other hand, the resolvent identity gives
\begin{align*}
\pppg x^{-\k \s_1} \partial_x^{\b_1} R_0(z) \partial_x^{\b_2} \pppg x^{-\k \s_2}
& = \pppg x^{-\k \s_1} \partial_x^{\b_1} \Rp(z) \partial_x^{\b_2} \pppg x^{-\k \s_2}\\
& - \pppg x^{-\k \s_1}\partial_x^{\b_1} \Rp(z) \Po R_0(z) \partial_x^{\b_2} \pppg x^{-\k \s_2}.
\end{align*}
Let $\tilde \s_1,\tilde \s_2 \in \big[0,\frac d 2 \big)$ be such that $\k \tilde \s_1 + \k \tilde \s_2 < \rho_G$. We have 
\begin{multline*}
\nr{\pppg x^{-\k \s_1}\partial_x^{\b_1} \Rp(z) \Po R_0(z) \partial_x^{\b_2} \pppg x^{-\k \s_2}}\\
\lesssim 
\nr{\pppg x^{-\k \s_1}\partial_x^{\b_1} \Rp(z) \nabla \pppg x^{-\k \tilde \s_2}}
\nr{\pppg x^{-\k \tilde \s_1} \nabla R_0(z) \partial_x^{\b_2} \pppg x^{-\k \s_2}}.
\end{multline*}
By \eqref{equiv-Ac-Ro} and Proposition \ref{prop-res-Rp} we obtain 
\begin{multline} \label{estim-Ro}
\nr{\pppg x^{-\k \s_1} \partial_x^{\b_1} R_0(z) \partial_x^{\b_2} \pppg x^{-\k \s_2}} \\
\lesssim \left( 1 + \abs{z}^{\frac {\abs{\b_1} + \abs{\b_2} + \s_1 + \s_2} 2 - 1}\right) 
+ \left( 1 + \abs z^{\frac {\abs{\b_1} + \s_1 + \tilde \s_2 - 1}2} \right) \nr{\pppg x^{- \k \tilde \s_1} \nabla R_0(z) \partial_x^{\b_2} \pppg x^{-\k \s_2}}.
\end{multline}
We first choose $\tilde \s_2 \in (0,1)$ and $\tilde \s_1 = 0$. We apply this estimate with $\b_1 = \b_2$ and $\s_1 = \s_2$ on the one hand, with $\b_1= 0$ and $\s_1 = 0$ on the other hand. This gives 
\begin{multline*} 
\nr{\pppg x^{-\k \s_2} \partial_x^{\b_2} R_0(z) \partial_x^{\b_2} \pppg x^{-\k \s_2}} \\
\lesssim \left( 1 + \abs{z}^{\abs{\b_2} + \s_2 - 1}\right) 
+ \left( 1 + \abs z^{\frac {\abs{\b_2} + \s_2  - 1}2} \right) \nr{\nabla R_0(z) \partial_x^{\b_2} \pppg x^{-\k \s_2}}
\end{multline*}
and
\[
\nr{R_0(z)  \partial_x^{\b_2} \pppg x^{-\k \s_2}} \lesssim \left( 1 + \abs{z}^{\frac {\abs{\b_2} + \s_2} 2 - 1}\right) 
+ \left( 1 + \abs z^{\frac {\tilde \s_2 - 1}2} \right) \nr{\nabla R_0(z) \partial_x^{\b_2} \pppg x^{-\k \s_2}}.
\]
Then \eqref{estim-nabla-Ro} gives
\begin{align*}
\nr{\nabla R_0(z) \partial_x^{\b_2} \pppg x^{-\k \s_2}}
& \lesssim \left( 1 + \abs z^{\frac {\abs {\b_2} + \s_2 - 1}2} \right) 
+ \left( 1 + \abs z^{\frac {\abs{\b_2} + \s_2 - 1}4} \right)  \nr{\nabla R_0(z) \partial_x^{\b_2} \pppg x^{-\k \s_2}}^{\frac 12}\\
& + \abs z^{\frac {\tilde \s_2}2} \nr{\nabla R_0(z) \partial_x^{\b_2} \pppg x^{-\k \s_2}}.
\end{align*}
For $z$ small enough this gives 
\begin{equation} \label{estim-nabla-Ro-bis}
\nr{\nabla R_0(z) \partial_x^{\b_2} \pppg x^{-\k \s_2}} \lesssim \left( 1 + \abs z^{\frac {\abs {\b_2} + \s_2 - 1}2} \right).
\end{equation}

Now we turn to the proof of
\begin{equation} \label{estim-Ro-2}
\nr{\pppg x^{-\k s_1} \partial_x^{\b_x} R_0(z) \partial_x^{\tilde \b_x} \pppg x^{-\k s_2}} \lesssim \left( 1 + \abs z^{\frac {\abs {\b_x} + |\tilde \b_x| + s_1 + s_2} 2 - 1} \right).
\end{equation}
With \eqref{equiv-Ac-Ro}, this will conclude the proof of the lemma.
Notice that it is enough to prove \eqref{estim-Ro-2} when $s_1 \leq 2 - \abs{\b_x}$. Indeed, the right-hand side does not really depend on $s_1 \geq 2 - \abs {\b_x}$, so if \eqref{estim-Ro-2} is proved for $s_1 = 2 - \abs{\b_x}$ it remains true for greater values of $s_1$. Similarly, it is enough to consider the case $|\tilde \b_x| + s_2 \leq 2$.

First assume that $\abs{\b_x} + s_1 \leq 1$. Then \eqref{estim-Ro-2} follows from \eqref{estim-Ro} applied with $\tilde \s_1 = 0$ and $\tilde \s_2 = \max(0,|\tilde \b_x| + s_2 -1)$. Then for $\abs{\b_x} + s_1 \in [1,2]$ we can apply \eqref{estim-Ro} with $\tilde \s_2 = 0$ and $\tilde \s_1 = \abs{\b_x} + s_1 - 1 \in [0,1]$. The proof is complete.
\end{proof}

\begin{lemma} \label{lemma-prop-res-P0}
If $\b_t = 0$ then the result of Proposition \ref{prop-res-general} holds with $\AAc$ replaced by $\tAAc$.
\end{lemma}

\begin{proof}
We begin with the case $m = 0$. The resolvent identity between $\tAAc$ and $\AAp$ reads 
\begin{equation} \label{res-id-2}
\begin{aligned}
(\tAAc-z)\inv - (\AAp-z)\inv
& = - (\AAp-z)\inv \Pc_0 (\tAAc-z)\inv\\
& = - (\tAAc-z)\inv \Pc_0 (\AAp-z)\inv.
\end{aligned}
\end{equation}
We can write
\[
\Pc_0 = - \sum_{1\leq j,k\leq d} \tilde \Th_{e_j} \begin{pmatrix} G_{0;j,k}(x) & 0 \\ 0 & 0 \end{pmatrix} \Th_0^{e_k},
\] 
where $(e_1,\dots,e_d)$ is the canonical basis in $\R^d$. For $\s_1,\s_2 \in \big[0,\frac d 2 \big)$ such that $\k \s_1 + \k \s_2 < \rho_G$ we obtain by Lemma \ref{lem-estim-P0}
\[
\nr{\Thbxo \big( (\tAAc-z)\inv - (\AAp-z)\inv \big)}_{\mathscr{L}(\LL^{\k s_2},\LL^{-\k s_1})} \lesssim \left( 1 + \abs z^{\frac {\abs{\b_x} + s_1 + \s_2 - 1}2} \right)  \left( 1 + \abs z^{\frac {\s_1 + s_2 - 1}2} \right).
\]
If $\abs{\b_x} + s_1 \geq 1$ and $s_2 \geq 1$ we can apply this inequality with $\s_1 = \s_2 = 0$ to conclude. If $\abs{\b_x} + s_1 \geq 1$ we can take $\s_2 = 0$ and $\s_1 = \abs{\b_x} + s_1 - 1 + \y$. If $s_2 \geq 1$ we can take $\s_1 = 0$ and $\s_2 = s_2 - 1 + \y$. Finally, if $\abs{\b_x} + s_1 < 1$ (then $\abs {\b_x} = 0$) and $s_2 < 1$ we choose $\s_1 \in [0,1-s_2]$ and $\s_2 \in [0,1-s_1]$ in such a way that
\[
\s_1 + \s_2 = \min(2 - s_1 - s_2, \y).
\]
This conclude the case $m = 0$.

Then we proceed by induction on $m$. With \eqref{res-id-2} we can check that 
\[
(\tAAc-z)^{-1-m} = \left( 1 - (\tAAc-z)\inv \Pc_0 \right) (\AAp-z)\inv (\tAAc-z)^{-m},
\]
and
\[
(\tAAc - z)^{-m} =  (\AAp - z)^{-m} - \sum_{k=1}^{m}  (\AAp-z)^{-k} \Pc_0 (\tAAc-z)^{k-m-1}.
\]
This gives
\begin{equation*}
(\tAAc-z)^{-1-m} = \left( 1 - (\tAAc-z)\inv \Pc_0 \right) \left((\AAp-z)^{-1-m} - \sum_{k=1}^{m}  (\AAp-z)^{-1-k} \Pc_0 (\tAAc-z)^{k-m-1}\right),
\end{equation*}
hence
\begin{equation} \label{dec-AAc-m}
\begin{aligned}
(\tAAc-z)^{-1-m} - (\AAp-z)^{-1-m} 
& = - (\tAAc-z)\inv \Pc_0 (\AAp-z)^{-1-m}\\
& - \left( 1 - (\tAAc-z)\inv \Pc_0 \right)  \sum_{k=1}^{m}  (\AAp-z)^{-1-k} \Pc_0 (\tAAc-z)^{k-m-1}.
\end{aligned}
\end{equation}
The interest of this decomposition is that we only have factors for which we can use the inductive assumption. We choose $k \in \Ii 1 m$ and estimate
\[
T_{m,k} = (\tAAc-z)\inv \Pc_0 (\AAp-z)^{-1-k} \Pc_0 (\tAAc-z)^{k-m-1}.
\]
We have 
\[
\nr{\Thbxo T_{m,k}}_{\mathscr{L}(\LL^{\k s_2},\LL^{-\k s_1})} \lesssim 
\left( 1 + \abs z^{\frac {\abs {\b_x} + s_1 + \s_2 - 1} 2} \right)
\left( 1 + \abs z^{\frac {\s_1 + \tilde \s_2} 2 -k} \right)
\left( 1 + \abs z^{\frac {\tilde \s_1 + s_2 - 1} 2 -m+k} \right),
\]
where $\s_1, \s_2,\tilde \s_1,\tilde \s_2 \in \big[0,\frac d 2)$ are such that $\k \s_1 + \k \s_2 < \rho_G$ and $\k \tilde \s_1 + \k \tilde \s_2 < \rho_G$. Then we play the same game as above, except that we have four parameters to choose.

Assume that $\abs {\b_x} + s_1 \geq 2k + 1$. Then we can take $\s_2 = 0$, $\s_1 = 2k$, $\tilde \s_2 = 0$ and $\tilde \s_1 = \abs {\b_x} + s_1 - 1 - 2k + \y$. Similarly, if $s_2 \geq 2(m+1) - 1$ we take $\tilde \s_1 = 0$, $\tilde \s_2 = 2k$, $\s_1 = 0$ and $\s_2 = s_2 - 2(m+1) + 1 + \y$.

Now assume that $\abs {\b_x} + s_1 < 2k + 1$ and $s_2 < 2(m+1) - 1$. If $\abs {\b_x} + s_1 < 1$ (then $\abs {\b_x} = 0$) then we take $\s_2 = \min(1-s_1,\y)$ and $\s_1 = \y - \s_2$. If $\abs {\b_x} + s_1 > 1$ then we take $\s_2 = 0$ and $\s_1 = \abs{\b_x} + s_1 - 1 + \y$. If $s_2 < 1$ we take $\tilde \s_1 = \min(1-s_2,\y)$ and $\tilde \s_2 = \y - \tilde \s_1$. Finally, if $s_2 > 1$ then we take $\tilde \s_1 = 0$ and $\tilde \s_2 = \s_2 - 1 + \y$. We can check that in any case we have 
\[
\nr{\Thbxo T_{m,k}}_{\mathscr{L}(\LL^{\k s_2},\LL^{-\k s_1})} \lesssim 1 + \abs z^{\frac {\abs{\b_x} + s_1 + s_2 + \y}2 - (1+m)}.
\]
The other terms in \eqref{dec-AAc-m} are estimated similarly, and the proof is complete. 
\end{proof}

\begin{remark}
With Lemma \ref{lem-estim-P0} and Lemma \ref{lem-prop-res-bt} applied with $\ao = \wo = 0$ we obtain 
\begin{equation} \label{estim-sans-Thbb}
\nr{(\tAAc-z)^{-1-m}}_{\mathscr{L}(\LL^{\k s_2},\LL^{-\k s_1})} \leq C \left( 1 + \abs z^{\frac {s_1 + s_2}2 - (1 + m)}\right).
\end{equation}
\end{remark}

It remains to add the contribution of $\Dc_0$. We begin with an estimate of the powers of $(\AAc -z)\inv$.

\begin{lemma} \label{lem-res-D0-single}
Let $s_1,s_2 \in \big[0,\frac d 2)$ and $\k > 1$. Then there exists $C \geq 0$ such that for $z \in \C_+$ with $\abs z \leq 1$ we have 
\[
\nr{(\AAc-z)^{-1-m}}_{\mathscr{L}(\LL^{\k s_2},\LL^{-\k s_1})} \leq C \left( 1 + \abs z ^{\frac {s_1 + s_2} 2 - (1+m)} \right).
\]
\end{lemma}

\begin{proof}
The resolvent identity between $\AAc$ and $\tAAc$ reads
\begin{equation} \label{res-id}
(\AAc-z)\inv = (\tAAp-z)\inv - (\tAAc-z)\inv \Dc_0 (\AAc-z)\inv.
\end{equation}
We can apply Remark \ref{rem-colonne-droite} to the operator $\Dc_0$. Moreover its coefficients decay according to \eqref{hyp-vanishing}. Thus, if $\s_1,\s_2 \in \big[ 0, \frac d 2 \big)$ and $\tilde \k > 1$ are such that $\tilde \k \s_1 + \tilde \k \s_2 < \rho_a$, we have by Lemma \ref{lem-estim-P0} and \eqref{estim-sans-Thbb}
\begin{eqnarray} \label{estim-tAAc-sans-Thbb}
\lefteqn{\nr{\pppg x^{-\k s_1} (\tAAc -z)\inv \pppg x^{-\k s_2}}}\\
\nonumber
&& \leq \nr{\pppg x^{-\k s_1} (\AAp -z)\inv \pppg x^{-\k s_2}}\\
&& \quad + \abs z \nr{\pppg x^{-\k s_1} (\AAp-z)\inv \pppg x^{-\tilde \k \s_2}} \nr{ \pppg x^{-\tilde \k \s_1} (\tAAc -z)\inv \pppg x^{-\k s_2}}\\
\nonumber
&& \lesssim \left( 1 + \abs{z}^{\frac {s_1 + s_2} 2 - 1} \right) + \left( \abs z + \abs{z}^{\frac {s_1 + \s_2} 2} \right) \nr{ \pppg x^{-\k \s_1}(\tAAc -z)\inv \pppg x^{-\k s_2}}.
\end{eqnarray}
If $s_1 = 0$ we apply this inequality with $\s_1 = 0$ and $\s_2 > 0$. This gives the required estimate for $z$ small enough (which is enough since we know that the resolvent is uniformly bounded outside some neighborhood of 0). Then if $s_1 \leq 2$ we apply \eqref{estim-tAAc-sans-Thbb} with $\s_1 = 0$ and $\s_2 = \max(0,s_2 -2)$ and get the same conclusion. Finally if $s_1 > 2$ we simply take $\s_1 = \s_2 = 0$, which concludes the case $m=0$.

Then we proceed by induction on $m$. With \eqref{res-id} we can check that
\begin{equation} \label{res-id-3}
(\AAc - z)^{-1-m} - (\tAAc - z)^{-1-m} = - \sum_{k=0}^m  (\tAAc-z)^{-1-k} \Dc_0 (\AAc-z)^{-1-m+k}.
\end{equation}
For $m \in \N$ and $k \in \Ii 0 m$ we set
\begin{equation} \label{def-Tmk}
T_{m,k}(z) =  (\tAAc - z)^{-1-k} \Dc_0 (\AAc-z)^{-1-m+k}.
\end{equation}
If $k = m$ we obtain by Remark \ref{rem-colonne-droite} 
\begin{multline*}
\nr{\pppg x^{-\k s_1} T_{m,k}(z) \pppg x^{-\k s_2}}
\lesssim \nr{\pppg x^{-\k s_1} (\tAAc-z)^{-1-m} \pppg x^{-\k s_2}}\\
+ \abs z \nr{\pppg x^{-\k s_1} (\tAAc-z)^{-1-m} \pppg x^{-\k \s_2}} \nr{\pppg x^{-\k \s_1} (\AAc-z)^{-1}  \pppg x^{-\k s_2}},
\end{multline*}
where, again, $\s_1,\s_2 \in \big[ 0, \frac d 2 \big)$ are such that $\k \s_1 + \k\s_2 < \rho_a$. Using the inductive assumption for the last factor, \eqref{estim-sans-Thbb} for the others, and choosing $\s_1$ and $\s_2$ suitably as above, we obtain
\begin{equation} \label{estim-Tmk}
\nr{\pppg x^{-\k s_1} T_{m,k}(z) \pppg x^{-\k s_2}} \lesssim 1 + \abs z^{\frac {s_1 + s_2} 2 - (1+m)}.
\end{equation}
For $k \in \Ii 0 {m-1}$ we use Remark \ref{rem-colonne-droite} and obtain 
\begin{eqnarray} \label{dec-Tmk}
\lefteqn{\nr{\pppg x^{-\k s_1} T_{m,k} \pppg x^{-\k s_2}}}\\
\nonumber
&& \leq \nr{\pppg x^{-\k s_1} (\tAAc-z)^{-1-k} \pppg x^{-\k \tilde \s_2}} \nr{\pppg x^{-\k \tilde \s_1} (\AAc-z)^{-m+k}  \pppg x^{-\k s_2}}\\
\nonumber
&& + \abs z \nr{\pppg x^{-\k s_1} (\tAAc-z)^{-1-k} \pppg x^{-\k \s_2}} \nr{\pppg x^{-\k \s_1} (\AAc-z)^{-1-m+k}  \pppg x^{-\k s_2}},
\end{eqnarray}
where $\tilde \s_1,\tilde \s_2,\s_1,\s_2 \in \big[0,\frac d 2\big)$ are such that $\k\tilde \s_1 + \k\tilde \s_2 \leq \rho_a$ and $\k \s_1 + \k \s_2 \leq \rho_a$. We proceed as above to obtain \eqref{estim-Tmk} if $k \in \Ii 1 {m-1}$. For $k = 0$ we get 
\begin{multline*}
\nr{\pppg x^{-\k s_1} T_{m,0}(z) \pppg x^{-\k s_2}}\\
\lesssim \left(1 + \abs z^{\frac {s_1 + s_2} 2 - (1+m)} \right) + \left( \abs z + \abs z^{\frac {s_1 + \s_2}2} \right) \nr{\pppg x^{-\k \s_1} (\AAc-z)^{-1-m}  \pppg x^{-\k s_2}}.
\end{multline*}
Finally,
\begin{multline*}
\nr{\pppg x^{-\k s_1} (\AAc-z)^{-1-m} \pppg x^{-\k s_2}}\\
\lesssim \left(1 + \abs z^{\frac {s_1 + s_2} 2 - (1+m)} \right) +  \left( \abs z + \abs z^{\frac {s_1 + \s_2}2} \right) \nr{\pppg x^{-\k \s_1} (\AAc-z)^{-1-m}  \pppg x^{-\k s_2}}.
\end{multline*}
As for the case $m = 0$, we conclude with $\s_1 = 0$ and $\s_2 > 0$ if $s_1 = 0$ and then with $\s_1 = 0$ and $\s_2 = \max(0,s_2-2(m+1))$ if $s_1 \leq 2$. Then we proceed by induction on the integer part of $\frac {s_1} 2$. If $s_1 \in (2k,2(k+1)]$ for some $k \geq 1$, then we choose $\s_2 = 0$ and $\s_1 = s_1 - 2 \in (2(k-1),2k]$ to conclude the proof.
\end{proof}

Finally the following lemma will conclude the proof of Proposition \ref{prop-res-general}.

\begin{lemma} \label{lemma-prop-res-D0}
The result of Proposition \ref{prop-res-general} holds if $\b_t = 0$.
\end{lemma}

\begin{proof}
We start again from \eqref{res-id-3} and use the notation \eqref{def-Tmk}. We consider the case $k \in \Ii 0 {m-1}$. By an estimate analogous to \eqref{dec-Tmk} we obtain 
\begin{align*} 
\nr{\pppg x^{-\k s_1} \Thbxo T_{m,k} \pppg x^{-\k s_2}}
& \leq \left( 1 + \abs z^{\frac {\abs {\b_x} + s_1 + \tilde \s_2} 2 - (1+k)} \right) \left( 1 + \abs z^{\frac {\tilde \s_1 + s_2} 2 - (m-k)} \right)  \\
& \leq \abs z \left( 1 + \abs z^{\frac {\abs {\b_x} + s_1 + \s_2} 2 - (1+k)} \right) \left( 1 + \abs z^{\frac {\s_1 + s_2} 2 - (1 + m-k)} \right),
\end{align*}
where $\tilde \s_1,\tilde \s_2,\s_1,\s_2 \in \big[0,\frac d 2\big)$ are such that $\k\tilde \s_1 + \k\tilde \s_2 \leq \rho_a$ and $\k \s_1 + \k \s_2 \leq \rho_a$. Choosing suitably these coefficients in the same spirit as above we get the estimates for the contributions of $T_{m,k}(z)$ for $k \in \Ii 0 {m-1}$. The case $k = m$ is similar, and the proof is complete.
\end{proof}

\subsection{Energy decay}

In this final paragraph we use the resolvent estimates of Proposition \ref{prop-res-general} to prove Theorem \ref{th-diff-decay}. 
We recall from \cite{MallougRo} the following lemma. See also \cite{Dewez}.

\begin{lemma} \label{lem-res-time}
Let $\Kc$ be a Hilbert space and let $I$ be an open bounded interval of $\R$. Let $\n \geq 0$, $\nu_0 > \n$ and $C \geq 0$.
Let $\f \in C_0^\infty(I,\Kc)$ and $\p \in C^\infty(I,\C)$. Assume that for $m \in \N$ with $m \leq \n_0 + 1$ and $\t \in I$ we have  
\[
\nr{\f^{(m)}(\t)}_\Kc \leq C \left( 1 + \abs \t^{\n_0 - 1 -m}\right),
\] 
\[
\abs{\p^{(m)}(\t)} \leq C \quad \text{and} \quad \abs{\p'(\t)} \geq \frac 1 C.
\]
Then there exists $c \geq 0$ which only depends on $I$, $\n$, $\n_0$ and $C$ such that for all $t \geq 0$ we have 
\[
\nr{\int_{I} e^{-it\p(\t)} \f(\t)\, d\t}_{\Kc} \leq c \pppg t^{-\n} \exp\left(t \sup_I \Im (\p) \right).
\]
\end{lemma}

Now we can finish the proof of Theorem \ref{th-diff-decay}.

\begin{proof} [Proof of Theorem \ref{th-diff-decay}]
It is enough to prove the result for $\k$ close to 1, so without loss of generality we can assume that \eqref{hyp-eta} holds. By density it is enough to prove the result for $F \in \Sc \times \Sc$. Let $\m \in (0,1]$. By Proposition \ref{prop-eAAc-Ilow} it is enough to estimate the difference between $\Ilow(t)$ and $\Ilowp(t)$. We recall that $\Ilow(t)$, $\Ilowp(t)$ and $\G_\m$ were defined in \eqref{def-Ilow}, \eqref{def-Ilowp} and \eqref{def-Gamma}, respectively. We have 
\begin{multline*}
\Thbxbt (\Ilow(t) - \Ilowp(t)) F\\
= \frac 1 {2i\pi} \int_{\t} \h(\t) e^{-it(\t + i \th_\m(\t))} \Thbxbt \big((\AAc-(\t + i \th_\m(\t)))\inv - (\AAp-(\t +i\th_\m(\t)))\inv \big) F \, d\t.
\end{multline*}
We can assume that the derivatives of $\th_\m$ are uniform in $\m \in (0,1]$. Then, by Lemma \ref{lem-res-time} and the estimates of Proposition \ref{prop-res-general} (with $\eta$ replaced by $\tilde \eta > \eta$ which still satisfies \eqref{hyp-eta}) there exists $c \geq 0$ which does not depend on $F \in \Sc \times \Sc$, $\m \in (0,1]$ or $t \geq 0$ such that  
\[
\nr{\Thbxbt (\Ilow(t) - \Ilowp(t)) F}_{\LL^{-\k s_1}} \leq c e^{t\m} \pppg t ^{-\b_t - \frac {\abs {\b_x}} 2 - \frac {s_1 + s_2} 2 - \frac \eta 2} \nr{F}_{\LL^{\k s_2}}.
\]
Then we let $\m$ go to 0, and the conclusion follows.
\end{proof}


\begin{thebibliography}{{Wak}17}

\bibitem[AIK15]{AlouiIbKh15}
L.~{Aloui}, S.~{Ibrahim}, and M.~{Khenissi}.
\newblock {Energy decay for linear dissipative wave equations in exterior
  domains.}
\newblock {\em {J. Differ. Equations}}, 259(5):2061--2079, 2015.

\bibitem[AK02]{alouik02}
L.~Aloui and M.~Khenissi.
\newblock Stabilisation pour l'\'equation des ondes dans un domaine
  ext\'erieur.
\newblock {\em Rev. Math. Iberoamericana}, 18:1--16, 2002.

\bibitem[All02]{allaire}
G.~Allaire.
\newblock {\em Shape Optimization by the Homogenization Method}.
\newblock Springer-Verlag, New York, 2002.

\bibitem[BH12]{bonyh12}
J.-F. Bony and D.~H\"afner.
\newblock {Local Energy Decay for Several Evolution Equations on Asymptotically
  Euclidean Manifolds.}
\newblock {\em Annales Scientifiques de l' \'Ecole Normale Sup\'erieure},
  45(2):311--335, 2012.

\bibitem[BJ16]{BurqJo}
N.~Burq and R.~Joly.
\newblock Exponential decay for the damped wave equation in unbounded domains.
\newblock {\em Communications in Contemporary Mathematics}, 18(6), 2016.

\bibitem[BLP78]{BensoussanLionsPapanicolaou}
A.~Bensoussan, J.-L. Lions, and G.~Papanicolaou.
\newblock {\em Asymptotic Analysis for Periodic Structures}.
\newblock Studies in Mathematics and Its Applications 5. Elsevier Science Ltd,
  1978.

\bibitem[BLR92]{bardoslr92}
C.~{Bardos}, G.~{Lebeau}, and J.~{Rauch}.
\newblock {Sharp sufficient conditions for the observation, control, and
  stabilization of waves from the boundary.}
\newblock {\em {SIAM J. Control Optim.}}, 30(5):1024--1065, 1992.

\bibitem[Bou11]{bouclet11}
J.-M. Bouclet.
\newblock Low frequency estimates and local energy decay for asymptotically
  {E}uclidean laplacians.
\newblock {\em Comm. Part. Diff. Equations}, 36:1239--1286, 2011.

\bibitem[BR14]{boucletr14}
J.-M. Bouclet and J.~Royer.
\newblock Local energy decay for the damped wave equation.
\newblock {\em Jour. Func. Anal.}, 266(2):4538--4615, 2014.

\bibitem[Bur98]{burq98}
N.~Burq.
\newblock {D\'ecroissance de l'\'energie locale de l'\'equation des ondes pour
  le probl\`eme ext\'erieur et absence de r\'esonance au voisinage du r\'eel.}
\newblock {\em Acta Math.}, 180(1):1--29, 1998.

\bibitem[CH04]{Chill-Ha-04}
R.~{Chill} and A.~{Haraux}.
\newblock {An optimal estimate for the time singular limit of an abstract wave
  equation.}
\newblock {\em {Funkc. Ekvacioj, Ser. Int.}}, 47(2):277--290, 2004.

\bibitem[COV02]{ConcaOrVa02}
C.~{Conca}, R.~{Orive}, and M.~{Vanninathan}.
\newblock {Bloch approximation in homogenization and applications.}
\newblock {\em {SIAM J. Math. Anal.}}, 33(5):1166--1198, 2002.

\bibitem[CV97]{ConcaVa97}
C.~{Conca} and M.~{Vanninathan}.
\newblock {Homogenization of periodic structures via Bloch decomposition.}
\newblock {\em {SIAM J. Appl. Math.}}, 57(6):1639--1659, 1997.

\bibitem[Dew16]{Dewez}
F.~Dewez.
\newblock Asymptotic estimates of oscillatory integrals with general phase and
  singular amplitude: Applications to dispersive equations.
\newblock 2016.
\newblock Preprint. Arxiv 1507.00883.

\bibitem[HO04]{hosonoo04}
T.~{Hosono} and T.~{Ogawa}.
\newblock {Large time behavior and $L^{p}$-$L^{q}$ estimate of solutions of
  2-dimensional nonlinear damped wave equations.}
\newblock {\em {J. Differ. Equations}}, 203(1):82--118, 2004.

\bibitem[{Ike}02]{Ikehata02}
R.~{Ikehata}.
\newblock {Diffusion phenomenon for linear dissipative wave equations in an
  exterior domain.}
\newblock {\em {J. Differ. Equations}}, 186(2):633--651, 2002.

\bibitem[ITY13]{IkehataToYo13}
R.~Ikehata, G.~Todorova, and B.~Yordanov.
\newblock Optimal decay rate of the energy for wave equations with critical
  potential.
\newblock {\em J. Math. Soc. Japan}, 65(1):183--236, 2013.

\bibitem[Kat80]{kato}
T.~Kato.
\newblock {\em Perturbation Theory for linear operators}.
\newblock Classics in Mathematics. Springer, second edition, 1980.

\bibitem[Khe03]{khenissi03}
M.~Khenissi.
\newblock {\'E}quation des ondes amorties dans un domaine ext\'erieur.
\newblock {\em Bull. Soc. Math. France}, 131(2):211--228, 2003.

\bibitem[Leb96]{lebeau96}
G.~Lebeau.
\newblock {\'E}quation des ondes amorties.
\newblock {In : A. Boutet de Monvel and V. Marchenko (editors), \emph{Algebraic
  and geometric methods in mathematical physics}, 73-109. Kluwer Academic
  Publishers}, 1996.

\bibitem[LR97]{lebeaur97}
G.~{Lebeau} and L.~{Robbiano}.
\newblock {Stabilisation de l'\'equation des ondes par le bord.}
\newblock {\em {Duke Math. J.}}, 86(3):465--491, 1997.

\bibitem[{Mat}76]{Matsumura76}
A.~{Matsumura}.
\newblock {On the asymptotic behavior of solutions of semi-linear wave
  equations.}
\newblock {\em {Publ. Res. Inst. Math. Sci.}}, 12:169--189, 1976.

\bibitem[Mel79]{melrose79}
R.~Melrose.
\newblock Singularities and energy decay in acoustical scattering.
\newblock {\em Duke Math. Journal}, 46(1):43--59, 1979.

\bibitem[MN03]{marcatin03}
P.~{Marcati} and K.~{Nishihara}.
\newblock {The $L^{p}$--$L^{q}$ estimates of solutions to one-dimensional
  damped wave equations and their application to the compressible flow through
  porous media.}
\newblock {\em {J. Differ. Equations}}, 191(2):445--469, 2003.

\bibitem[MR]{MallougRo}
M.~Malloug and J.~Royer.
\newblock Energy decay in a wave guide with damping at infinity.
\newblock Preprint arXiv:1606.02549.

\bibitem[MRS77]{morawetzrs77}
C.S. Morawetz, J.V. Ralston, and W.A. Strauss.
\newblock Decay of the solution of the wave equation outside non-trapping
  obstacles.
\newblock {\em Comm. on Pure and Applied Mathematics}, 30:447--508, 1977.

\bibitem[{Nar}04]{narazaki04}
T.~{Narazaki}.
\newblock {$L^p$-$L^q$ estimates for damped wave equations and their
  applications to semi-linear problem.}
\newblock {\em {J. Math. Soc. Japan}}, 56(2):585--626, 2004.

\bibitem[{Nis}03]{nishihara03}
K.~{Nishihara}.
\newblock {$L^p$-$L^q$ estimates of solutions to the damped wave equation in
  3-dimensional space and their application.}
\newblock {\em {Math. Z.}}, 244(3):631--649, 2003.

\bibitem[{Nis}16]{Nishiyama16}
H.~{Nishiyama}.
\newblock {Remarks on the asymptotic behavior of the solution to damped wave
  equations.}
\newblock {\em {J. Differ. Equations}}, 261(7):3893--3940, 2016.

\bibitem[OZ00]{OrtegaZu00}
J.~H. {Ortega} and E.~{Zuazua}.
\newblock {Large time behavior in $\mathbb R^N$ for linear parabolic equations
  with periodic coefficients.}
\newblock {\em {Asymptotic Anal.}}, 22(1):51--85, 2000.

\bibitem[OZP01]{OrivePaZu01}
R.~{Orive}, E.~{Zuazua}, and A.F. {Pazoto}.
\newblock {Asymptotic expansion for damped wave equations with periodic
  coefficients.}
\newblock {\em {Math. Models Methods Appl. Sci.}}, 11(7):1285--1310, 2001.

\bibitem[Roy]{royer-diss-wave-guide}
J.~Royer.
\newblock Local energy decay and diffusive phenomenon in a dissipative wave
  guide.
\newblock Preprint, arXiv:1601.05299.

\bibitem[Roy16]{royer-dld-energy-space}
J.~Royer.
\newblock Local decay for the damped wave equation in the energy space.
\newblock {\em Journal of the Institute of Mathematics of Jussieu}, 2016.
\newblock To appear, available online:
  \url{http://dx.doi.org/10.1017/S147474801600013X}.

\bibitem[RT74]{raucht74}
J.~Rauch and M.~Taylor.
\newblock Exponential decay of solutions to hyperbolic equations in bounded
  domains.
\newblock {\em Indiana Univ. Math. J.}, 24(1):79--86, 1974.

\bibitem[RTY10]{Radu-To-Yo-11}
P.~{Radu}, G.~{Todorova}, and B.~{Yordanov}.
\newblock {Decay estimates for wave equations with variable coefficients.}
\newblock {\em {Trans. Am. Math. Soc.}}, 362(5):2279--2299, 2010.

\bibitem[RTY16]{Radu-To-Yo-16}
P.~{Radu}, G.~{Todorova}, and B.~{Yordanov}.
\newblock The generalized diffusion phenomenon and applications.
\newblock {\em SIAM J. Math. Anal.}, 48(1):174--203, 2016.

\bibitem[SW16]{SobajimaWak16}
M.~{Sobajima} and Y.~{Wakasugi}.
\newblock {Diffusion phenomena for the wave equation with space-dependent
  damping in an exterior domain.}
\newblock {\em {J. Differ. Equations}}, 261(10):5690--5718, 2016.

\bibitem[Tar09]{tartar}
L.~Tartar.
\newblock {\em The General Theory of Homogenization}.
\newblock Lecture Notes of the Unione Matematica Italiana. 2009.

\bibitem[TY09]{TodorovaYo09}
G.~{Todorova} and B.~{Yordanov}.
\newblock {Weighted $L^2$-estimates for dissipative wave equations with
  variable coefficients.}
\newblock {\em {J. Differ. Equations}}, 246(12):4497--4518, 2009.

\bibitem[{Wak}14]{Wakasugi14}
Y.~{Wakasugi}.
\newblock {On diffusion phenomena for the linear wave equation with
  space-dependent damping.}
\newblock {\em {J. Hyperbolic Differ. Equ.}}, 11(4):795--819, 2014.

\bibitem[{Wak}17]{Wakasugi17}
Y.~{Wakasugi}.
\newblock {Scaling variables and asymptotic profiles for the semilinear damped
  wave equation with variable coefficients.}
\newblock {\em {J. Math. Anal. Appl.}}, 447(1):452--487, 2017.

\bibitem[Wun]{Wunsch}
J.~Wunsch.
\newblock Periodic damping gives polynomial energy decay.
\newblock {\em Math. Res. Lett.}
\newblock To appear.

\bibitem[Zwo12]{zworski}
M.~Zworski.
\newblock {\em Semiclassical Analysis}, volume 138 of {\em Graduate Studies in
  Mathematics}.
\newblock American Mathematical Society, 2012.

\end{thebibliography}
\end{document}